\documentclass[12pt]{iopartmod}
\usepackage{graphicx,amsmath,amssymb,dsfont,subfigure, color, amsthm}

\newcommand{\bra}[1]{\mbox{$\langle #1 |$}}
\newcommand{\ket}[1]{\mbox{$| #1 \rangle$}}

\newtheorem{lem}{Lemma}
\newtheorem{prop}{Proposition}
\newtheorem{cor}{Corollary}

\begin{document}

\title{Security Analysis of an Untrusted Source for Quantum Key Distribution: Passive Approach}
\author{Yi Zhao\footnote{Present address: California Institute of Technology, MC 12-33, Pasadena, CA, 91106, USA.}, Bing Qi, Hoi-Kwong Lo, and Li Qian}
\address{Center for Quantum Information and Quantum Control,\\
Department of Physics and Department of Electrical \& Computer
Engineering,\\ University of Toronto, Toronto, Ontario, M5S 3G4,
Canada.}

\ead{phy.zhao@utoronto.ca}

\begin{abstract}
We present a passive approach to the security analysis
of quantum key distribution (QKD) with an untrusted source. A complete proof of its unconditional security is also presented. This scheme has
significant advantages in real-life
implementations as it does not require fast optical switching or a quantum random number generator. The essential idea is to use a beam splitter to
 split each input pulse. We show that we can
characterize the source using a cross-estimate
technique without active routing of each pulse. We have derived analytical expressions for the passive
estimation scheme. Moreover, using simulations, we have considered
four real-life imperfections: Additional loss introduced by the
``plug \& play'' structure, inefficiency of the intensity monitor,
noise of the intensity monitor, and statistical fluctuation
introduced by finite data size. Our simulation results show that
the passive estimate of an untrusted source remains useful in practice,
despite these four imperfections. Also, we have performed
preliminary experiments, confirming the utility of our proposal
in real-life applications.
 Our proposal makes it
possible to implement the ``plug \& play'' QKD with
the security guaranteed, while keeping the implementation practical. 
\end{abstract}


\maketitle

\section{Introduction}

Quantum key distribution (QKD) provides a means of sharing a secret
key between two parties, a sender Alice and a receiver Bob, securely
in the presence of an eavesdropper,
Eve~\cite{BB84,Ekert91,Review:Encyclopedia}. The unconditional
security of QKD has been rigorously proved~\cite{SecurityProofs},
even when implemented with imperfect real-life
devices~\cite{GLLP,ILM}. Decoy state method was proposed \cite{Decoy:Hwang,Decoy:LoISIT,Decoy:LoPRL,Decoy:Practical,
Decoy:WangPRL,Decoy:WangPRA} and
experimentally demonstrated
\cite{Decoy:ZhaoPRL,Decoy:ZhaoISIT} as a means
to dramatically improve the performance of QKD with imperfect
real-life devices with unconditional security still guaranteed
\cite{GLLP,Decoy:LoPRL}.

A large class of QKD setups adopts the so-called ``plug \& play''
architecture~\cite{Proposal:PnP,ExpQKD:PnP_67km}. In this setup, Bob
sends strong pulses to Alice, who encodes her quantum information on
them and attenuates these pulses to quantum level before sending
them back to Bob. Both phase and polarization drifts are
intrinsically compensated, resulting in a very stable and relatively
low quantum bit error rate (QBER). These significant practical
advantages make the ``plug \& play'' very attractive. Indeed, most
current commercial QKD systems are based on this particular scheme
\cite{IDQ,MagiQ}.

The security of ``plug \& play'' QKD was a long-standing open
question. A major concern arises from the following fact: When Bob
sends strong classical pulses to Alice, Eve can freely manipulate
these pulses, or even replace them with her own sophisticatedly
prepared pulses. That is, the source is equivalently controlled by
Eve in the ``plug \& play'' architecture. In particular, it is no
longer correct to assume that the photon number distribution is
Poissonian, as is commonly assumed in standard security proof. This
is a major reason why standard security proofs such as GLLP
\cite{GLLP} does not appear to apply directly to the ``plug and
play'' scheme.

It might be tempting to apply the central limit theorem \cite{CentralLimit} to the current problem. That is, the photons contained in a pulse after heavy attenuation obeys a Gaussian distribution asymptotically. The central limit theorem was adopted in \cite{ThrHack:TrojanHorse}.

However, the central limit theorem does not apply to the situation that the current paper is addressing. The current paper, as well as a previous work \cite{Security:UntrustedSource}, does not rely on the central limit theorem and removes the assumption on the input photon number distribution. i.e., our analysis applies to sources with an arbitrary photon number distribution. For example, imagine a source that follows a dual-delta distribution (i.e., the pulses sent by the source  contain either $n_1$ or $n_2$ photons, where $n_1$ and $n_2$ are large and different integers). In this case, even if Alice applies heavy attenuation on the input pulses, the resulting photon number per pulse distribution would be the sum of two Gaussian distributions, which in general is not a Gaussian distribution.

The dual-delta distributed source is of significant practical meaning rather than a purely imaginary source. Consider the case of the Trojan horse attack \cite{ThrHack:TrojanHorse}: An eavesdropper occasionally sends bright pulse to Alice and splits the corresponding output signal from Alice. In this case, the input photon number per pulse distribution on Alice's side would have two peaks: one corresponds to the photon number of the authentic source, and one corresponds to the sum of the photon numbers from the two pulses (one from the authentic source and the other one from the eavesdropper's probing pulse). The security analysis that is based on the central-limit theorem (e.g., Ref. \cite{ThrHack:TrojanHorse}) may not be directly applicable to this case. However, the analysis proposed in \cite{Security:UntrustedSource} and in the current paper can analyze such case simply by defining an appropriate input photon number range of the untagged bits such that most input pulses are included. Note that, the input photon number range for the untagged bits is defined in the post-processing stage, during which Alice has already collected the photon number distribution of the samples.

\begin{figure}\center
  \includegraphics[width=10cm]{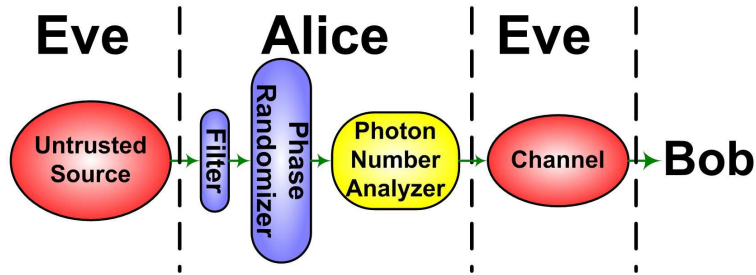}\\
  \caption{A general schematic of secure QKD with unknown and
  untrusted source. The Filter guarantees the single mode
  assumption. The Phase Randomizer guarantees the phase
  randomization assumption. The Photon Number Analyzer (PNA) estimates
  photon number distribution of the source. Various PNAs are shown
  in Figure \ref{FIG:Beamsplitter}.}\label{FIG:Setup_PNA}
\end{figure}

The unconditional security of ``plug \& play'' QKD
scheme has been recently proven in~\cite{Security:UntrustedSource}.
The basic idea is illustrated in Figure \ref{FIG:Setup_PNA}. A
Filter guarantees the single mode assumption. A Phase Randomizer
guarantees the phase randomization assumption. Note that for the state that is accessible to the eavesdropper, Alice's phase randomization is equivalent to a quantum non-demolition (QND) measurement of the photon numbers of the optical pulses. See \ref{App:PR} for details. Therefore, from now on, without loss of generality, we will assume that Alice's input signal is a classical mixture of Fock states and, similarly Alice's output signal is also a classical mixture of Fock states. A Photon Number
Analyzer (PNA) estimates photon number distribution of the source.
Detail of the PNA in \cite{Security:UntrustedSource} is shown in Figure
\ref{FIG:Beamsplitter}(a).

\begin{figure}\center
  \includegraphics[width=10cm]{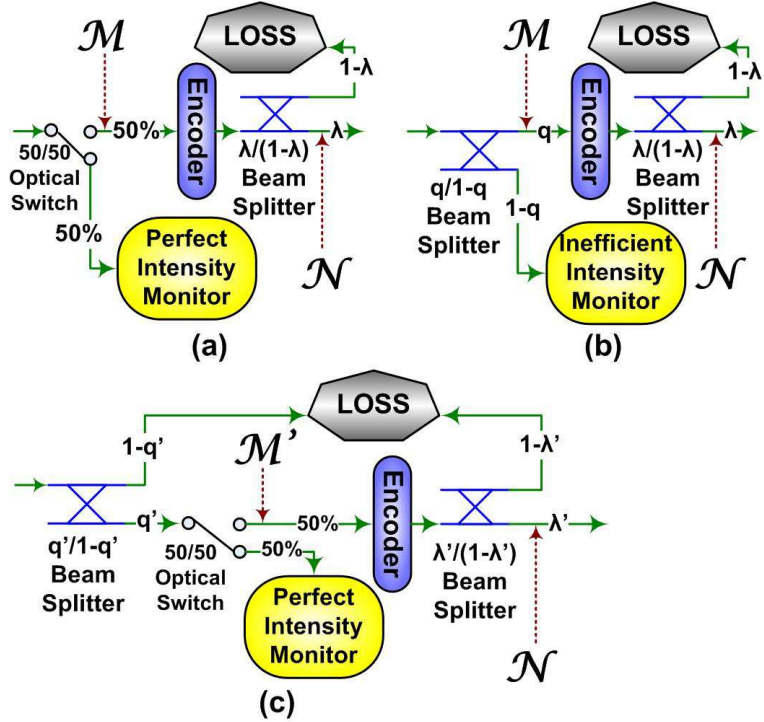}\\
  \caption{Different schemes to estimate photon number distribution.
  $\mathcal{M}$, $\mathcal{M'}$, and $\mathcal{N}$ are random
  variables for input photon number, virtual input photon number,
  and output photon number, respectively. All the internal loss of
  Alice is modeled as a $\lambda/1-\lambda$ beam splitter (in (a) and
  (b)) or a $\lambda'/1-\lambda'$ beam splitter (in (c)).
  (a) active scheme; (b) passive scheme; (c) hybrid scheme.
  $q'=\eta_\mathrm{IM}(1-q)$,
  where $\eta_\mathrm{IM}\le1$ is the efficiency of the imperfect
  intensity monitor. $\lambda' = q\lambda/q'$. Note that the scheme shown in
  (c) is a virtual set-up that has features from both the active scheme (a)
  and the passive scheme (b). The purpose of introducing this virtual scheme (c)
  is to bridge the active scheme (a) and the passive scheme (b).}\label{FIG:Beamsplitter}
\end{figure}

The analysis presented in \cite{Security:UntrustedSource} applies to
a general class of QKD with unknown and untrusted sources besides
``plug \& play'' QKD. For example, many QKD implementations use
pulsed laser diodes as the light source. These laser diodes are
turned on and off frequently to generate laser pulse sequence.
However, such laser pulses are not in coherent state
 and the photon number per pulse does not obey Poisson distribution
 \cite{Security:UntrustedSource}. Moreover, the go-and-return scheme
is also adopted by the recently proposed ground-satellite QKD
project \cite{ExpQKD:Ground_Satellite_Feasibility}, in which the
source is also equivalently unknown and untrusted.

\cite{Security:UntrustedSource} analyzes the photon number
distribution of an untrusted source in the following manner: Each
input pulse will be randomly routed to either an Encoder in Figure
\ref{FIG:Beamsplitter}(a) as a coding pulse, or a Perfect Intensity
Monitor in Figure \ref{FIG:Beamsplitter}(a) as a sampling pulse. The
photon numbers of each sampling pulse are individually measured by
the intensity monitor. In particular, one can obtain an estimate of
the fraction of coding pulses that has a photon number $m
\in[(1-\delta)M,(1+\delta)M]$ (here $\delta$ is a small positive real
number, and $M$ is a large positive integer. Both $\delta$ and $M$
are chosen by Alice and Bob). These bits are defined as ``untagged
bits''. The details of security analysis results of \cite{Security:UntrustedSource} are presented in \ref{App:Security}. We note that some security analyses about QKD with a fluctuating source have been reported recently \cite{Decoy:WangAPL,Decoy:WangInexactSource,Decoy:WangGeneralErrorSource, Decoy:WangNJP}.

It is challenging and inefficient to implement the scheme proposed
in~\cite{Security:UntrustedSource}, which is referred to as an active
scheme, for the following reasons: 1) The Optical Switch in Figure
\ref{FIG:Beamsplitter}(a) is an active component and requires
real-time control. The design and manufacture of the optical switch
and its controlling system can be very challenging in high-speed QKD
systems, which can operate as fast as 10 GHz~\cite{DPSK:200km}. 2) The random routing of optical pulses requires a high-speed sampling quantum random number generator (sampling QRNG), which does not yet exist for Gb/s systems. 3) 
The number of pulses sent to Bob is only a constant
fraction (say half) of the number of pulses generated by the source,
which means the key generation rate per pulse sent by the source is
reduced by that fraction.

Naturally, the optical switch can be replaced by a beam splitter,
which will passively split every input pulse, sending a portion into
the intensity monitor and the rest to the encoder. This is referred to
as a passive scheme. In this scheme, the sampling QRNG is not required.

A very recent work proposed some preliminary analysis on the passive
estimation of an untrusted source using inverse Bernoulli
transformation, and performed some  experimental tests
\cite{ExpQKD:PKU_Untrusted}. It is very encouraging to see that it is possible to prove the security of the passive estimate scheme for QKD with an untrusted source. As acknowledged by the authors of \cite{ExpQKD:PKU_Untrusted}, the inverse Bernoulli transformation is beyond the computational power of current computers, and the required photon number resolution is beyond the capabilities of practical photo diodes. Owing to the above challenges, the experimental data reported in
\cite{ExpQKD:PKU_Untrusted} were not analyzed by the analysis
proposed in the same paper.

In this paper, we propose a passive scheme to estimate the photon
number distribution of an untrusted source together with a complete
proof of its unconditional security. We show that the unconditional
security can still be guaranteed without routing each input optical
pulse individually. Our analysis provides both an analytical method to
calculate the final key rate and an explicit expression of the
confidence level. Moreover, we considered the inefficiency and
finite resolution of the intensity monitor, making our proposal
immediately applicable. In the numerical simulation, we considered
the additional loss introduced by the ``plug \& play'' structure and
the statistical fluctuation introduced by the finite data size. We
also gave examples of imperfect intensity monitors in the
simulation, in which a constant Gaussian noise is considered.

This paper is organized in the following way: in Section II, we will
propose a modified active estimate method; in Section III, we will
establish the equivalence between the modified active scheme
proposed in Section II and passive estimate scheme; in Section IV,
we will present a more efficient passive estimate protocol than the
one proposed in Section III; in Section V, we will present the
numerical simulation results of the protocol proposed in Section IV
and compare the efficiencies of active and passive estimates; in
Section VI, we will present a preliminary experiment based on our
proposed passive estimate protocol.

\section{Modified Active Estimate}

In \cite{Security:UntrustedSource}, it is shown that Alice can
randomly pick a fixed number of input pulses as sampling pulses, and
measure the number of untagged sampling bits. One can then estimate
the number of untagged coding bits.

We find that we can modify the scheme proposed in
\cite{Security:UntrustedSource} by drawing a
non-fixed number of input pulses as samples. A passive estimate can be
built on top of this modified active estimate scheme. Note that we
only modified the way to estimate the number of untagged coding
bits. Once the number of untagged coding bits is estimated, the
security analysis proposed in \cite{Security:UntrustedSource} is
still applicable to calculate the lower bound of secure key rate.

\begin{lem}\label{Lemma:ActiveConfidenceLevel}
Consider that $k$ pulses are sent to Alice from an unknown and
untrusted source, within which $V$ pulses are untagged. Alice
randomly assigns each bit as either a sampling bit or a coding bit
with equal probabilities (both are 1/2). In total, $V_\mathrm{s}$
sampling bits and $V_\mathrm{c}$ coding bits are untagged.
 The probability that $V_\mathrm{c}\le
V_\mathrm{s}-\epsilon k$ satisfies
\begin{equation}\label{Eq:SamplingCodingDeviation}
    P(V_\mathrm{c}\le V_\mathrm{s}-\epsilon k)\le
    \exp(-\frac{k\epsilon^2}{2})
\end{equation}
where $\epsilon$ is a small positive real number chosen by Alice and
Bob.

That is, Alice can conclude that $V_\mathrm{c}>
V_\mathrm{s}-\epsilon k$ with confidence level
\begin{equation}\label{EQ:Lemma1_ConfidenceLevel}
    \tau>1-\exp(-\frac{k\epsilon^2}{2})
\end{equation}

\end{lem}

\begin{proof}
See \ref{App:Confidence_Active}.
\end{proof}

Note that the right hand side of Equation
\eqref{Eq:SamplingCodingDeviation} is independent of $V$. This is
important because Alice does not know the exact value of $V$, while
Eve may know, and may even manipulate the value of $V$. Nonetheless,
the inequality suggested in Equation
\eqref{Eq:SamplingCodingDeviation} holds for any possible value of
$V$. Therefore, Alice can always estimate that the
$V_\mathrm{c}>V_\mathrm{s}-\epsilon k$ with confidence level
$\tau_\mathrm{a}\ge1-\exp(-\frac{k\epsilon^2}{2})$. Note that the
estimate given in Lemma \ref{Lemma:ActiveConfidenceLevel} is
actually quite good for us because we will mainly be interested in
the case where $V$ is close to $k$.

\begin{figure}[!t]\center
  \includegraphics[width=11cm]{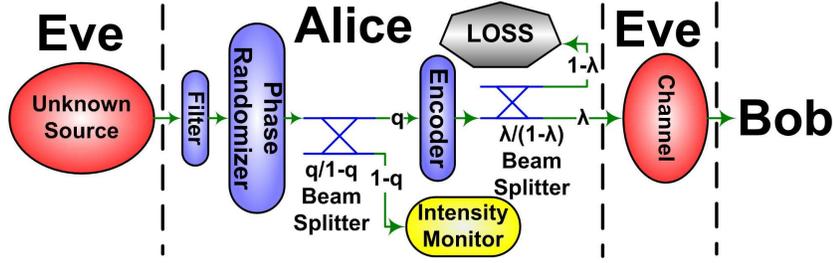}\\
  \caption{A schematic diagram of our proposed secure QKD scheme
  with passive estimate on an unknown and untrusted source.
  The Filter guarantees the single mode
  assumption,  and the $q/1-q$ Beam Splitter
  and the Intensity Monitor are used to passively
  estimate the photon number of input pulses.
  All the internal losses inside Alice's local
  lab is modeled as a $\lambda/(1-\lambda)$ beam
  splitter. That is, any input photon has $\lambda$
  probability to get encoded and sent from Alice to
  Bob, and $1-\lambda$ probability to be lost.}
 \label{FIG:Schematic}
\end{figure}

\section{From Active Estimate to Passive
Estimate}\label{Sec:ActiveToPassive}

The PNA of our proposed scheme is shown in Figure
\ref{FIG:Beamsplitter} (b) and the entire scheme is shown in Figure
\ref{FIG:Schematic}. We replaced the 50/50 Optical Switch in Figure
\ref{FIG:Beamsplitter} (a) by a $q/1-q$ Beam Splitter in Figure
\ref{FIG:Beamsplitter} (b). In this scheme, each input pulse is
passively split into two: One (defined as U pulse) is sent to the
encoder and transmitted to Bob, and the other (defined as L pulse)
is sent to the intensity monitor. The visualization of U/L pulses is
shown in Figure \ref{FIG:UL_Pulse}.

One may na\"{i}vely think that since the beam splitting ratio $q$ is
known, one can easily estimate the photon number of the U pulse from the
measurement result of photon number of the corresponding L pulse.
However, this is not true. Any input pulse, after the phase
randomization, is in a number state. Therefore, for a pair of U and L
pulses originating from the same input pulse, the total photon number
of the two pulses is an unknown constant. This restriction suggests
that we should not treat the photon numbers of such two pulses as
independent variables, and the random sampling theorem cannot be
directly applied.

\begin{figure}\center
  \includegraphics[width=8cm]{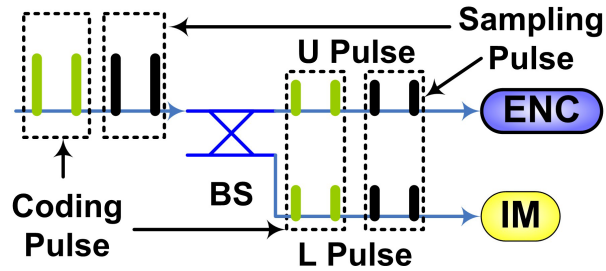}\\
  \caption{Visualization of different types of pulses. BS: Beam Splitter.
  ENC: Encoder. IM: Intensity Monitor. Each input pulse
  is randomly assigned as either a coding pulse or a sampling pulse. After
  entering the beam splitter, each pulse is split into a U pulse that enters
  the encoder, and an L pulse that enters the intensity monitor. As a result,
  there are four types of pulse: coding U pulse, coding L pulse, sampling U
  pulse, and sampling L pulse.}\label{FIG:UL_Pulse}
\end{figure}

To bridge the active scheme (in Figure \ref{FIG:Beamsplitter} (a))
and the passive scheme (in Figure \ref{FIG:Beamsplitter} (b)), we
introduce a virtual setup (in Figure \ref{FIG:Beamsplitter} (c)). We
call such a virtual set-up a ``hybrid'' scheme because it has
features from both the active and the passive schemes. 

We assume that the inefficiency of the intensity monitor can be
modeled as an additional loss \cite{ExpQKD:PKU_Untrusted}. In the passive scheme (Figure
\ref{FIG:Beamsplitter} (b)), assuming that the efficiency of the
intensity monitor is $\eta_\mathrm{IM}\le 1$, the probability that
an input photon is detected is
\begin{equation}\label{EQ:Virtual_Splitting_Ratio}
    q'=(1-q)\eta_\mathrm{IM}.
\end{equation}
Therefore, we could model the $q/1-q$ beam splitter and the
inefficient intensity monitor in Figure \ref{FIG:Beamsplitter} (b)
as a $q'/1-q'$ beam splitter and a perfect intensity monitor as in
Figure \ref{FIG:Beamsplitter} (c).

The above modification changes the probability that an input photon is sent to Bob. To ensure that an identical attenuation is applied to the coding
pulses in both the passive scheme (in Figure \ref{FIG:Beamsplitter}
(b)) and the hybrid scheme (in Figure \ref{FIG:Beamsplitter} (c)), we re-define the internal
transmittance in the virtual setup as
\begin{equation}\label{EQ:Virtual_Loss}
\lambda'=q\lambda/q'\le1
\end{equation}
For a given input photon number distribution, the output photon number distribution is determined by the internal loss \cite{Security:UntrustedSource}. Since the internal losses in the passive scheme and the hybrid scheme are identical, for a given input photon number distribution (which can be unknown), the output photon number distributions of the passive scheme is identical to that of the hybrid scheme. Moreover, the photon number distributions obtained by the intensity monitors are also identical for these two schemes.
 
Note that this virtual set-up is not actually used in an experiment,
but is purely for building the equivalence between the active and
the passive schemes.

By putting Equations
\eqref{EQ:Virtual_Splitting_Ratio} and \eqref{EQ:Virtual_Loss} together,
we have one constraint:
\begin{equation}\label{EQ:Restriction}
    \lambda'=\frac{q\lambda}{(1-q)\eta_\mathrm{IM}}\le1.
\end{equation}
This constraint is very easy to meet in an actual experiment as
$\lambda$ can be lower than $10^{-6}$ in a practical set-up
\cite{Security:UntrustedSource}, $q/(1-q)\le100$ in typical beam
splitters, and $\eta_\mathrm{IM}$ can be greater than 50\% in
commercial photo diodes \footnote{Several commercial high-speed
InGaAs photodiodes, including Thorlabs FGA04, JDSU EPM745 and Hamamatsu
G6854-01 are claimed to have conversion efficiency over $70\%$ at
1550nm.}.

The resolution of the intensity monitor is another important
imperfection. In a real experiment, the intensity monitor may indicate
a certain pulse contains $m'$ photons. Here we refer to $m'$ as
the \emph{measured} photon number in contrast to the \emph{actual}
photon number $m$. However, due to the noise and the inaccuracy of the
intensity monitor, this pulse may not contain exactly $m'$ photons.
To quantify this imperfection, we introduce a term ``the conservative
interval'' $\varsigma$. We then define $\underline{V}^\mathrm{L}$ as
the number of L pulses with measured photon number
$m'\in[(1-\delta)M+\varsigma,(1+\delta)M-\varsigma]$. One can
conclude that, with confidence level
$\tau_\mathrm{c}=1-c(\varsigma)$, the number of untagged L bits
$V^\mathrm{L}\ge \underline{V}^\mathrm{L}$. One can make
$c(\varsigma)$ arbitrarily close to 0 by choosing large enough
$\varsigma$ \footnote{The specific expression of $c(\varsigma)$ depends
on properties of a specific intensity monitor. Nonetheless, one can
always make $c(\varsigma)$ arbitrarily close to 0 by choosing a
large enough $\varsigma$. That is, $\forall\zeta>0$, we can always
find $\underline{\varsigma}\in[0,\delta M]$ such that for any
$\varsigma\ge\underline{\varsigma}$, we have $c(\varsigma)<\zeta$.
Note that $c(\delta M)=0$.}. The conservative interval is a statistical
property rather than an individual property. That is, for one
individual pulse, the probability that $|m-m'|>\varsigma$ can be
non-negligible.

In the virtual setup, input pulses are treated in the same manner as
in the active estimate scheme: Coding pulses are routed to the encoder
and then sent to Bob, while the sampling pulses are routed to the
perfect intensity monitor to measure their photon numbers. We can
use the measurement results of sampling pulses to estimate the
number of untagged bits in the coding pulses. Knowing the number of
untagged bits, one can easily calculate the upper and lower bounds
of the output photon number probabilities
\cite{Security:UntrustedSource}.

Since the passive scheme and the hybrid scheme share the same
source, the output photon number distribution is solely determined
by the internal loss. The internal transmittances for the coding bits
are the same ($q'\lambda'=q\lambda$) for both schemes. Therefore,
the upper and lower bounds of output photon number probabilities
estimated from the hybrid scheme are also valid for those of the
passive scheme.

\begin{cor}\label{Col:1}
 Consider $k$ pulses sent from an unknown and untrusted source
to Alice, where $k$ is a large positive integer. Alice randomly
assigns each input pulse as either a sampling pulse or a coding
pulse with equal probabilities. Define variables
$V_\mathrm{s}^\mathrm{L}$ and $V_\mathrm{c}^\mathrm{U}$ as the
number of untagged sampling L pulses and the number of untagged
coding U pulses, respectively. Here U pulses are defined as pulses
sent to the Encoder in Figure \ref{FIG:UL_Pulse}, and L pulses are
defined as pulses sent to the Intensity Monitor in Figure
\ref{FIG:UL_Pulse}. Alice can conclude that
$V_\mathrm{c}^\mathrm{U}> V_\mathrm{s}^\mathrm{L}-\epsilon_1 k$ with
confidence level $\tau_1 \ge 1-e^{-k\epsilon_1^2/2}$. Here
$\epsilon_1$ is a small positive real number chosen by Alice and Bob. To
calculate the upper and lower bounds of output photon number
probabilities, one should use equivalent internal transmittance
$\lambda'$, which is given in Equation \eqref{EQ:Restriction},
instead of actual internal transmittance $\lambda$.
\end{cor}

\begin{proof}
The sampling L pulses are sent to a perfect intensity monitor with a probability $q'=q\eta_\text{IM}$. If we apply the same transmittance $q'$ to the coding U pulses, we can consider that sampling L pulses and the coding U pulses as a group of pulses that go through the same attenuation, and we randomly assign each pulse in the group as either a sampling L pulse or a coding U pulse with equal probabilities. Therefore, one can conclude that $V_\mathrm{c}^\mathrm{U}> V_\mathrm{s}^\mathrm{L}-\epsilon_1 k$ with
confidence level $\tau_1 \ge 1-e^{-k\epsilon_1^2/2}$ by applying Lemma \ref{Lemma:ActiveConfidenceLevel}. Since the overall transmittance for the U pulses is $q\lambda$, the internal transmittance for the untagged coding U pulses should be considered as $\lambda'=q\lambda/q'$.
\end{proof}

There is no physical location (eg. between the Beam Splitter and the Encoder in Figure \ref{FIG:UL_Pulse}) where the U pulses see a transmittance of $q'$ in the passive scheme. The output photon number probabilities of the coding U pulses are analyzed in the following manner: The coding U pulses, after propagating through a virtual transmittance $q'$, contains $V_\text{c}^\text{U}$ untagged bits. These coding U pulses then propagate through another virtual transmittance $\lambda'$, and we can calculate the output photon number probabilities, which is identical to the output photon number probabilities generated by sending the coding U pulses through the real transmittance $q\lambda=q'\lambda'$. 

Note that, it is not clear to us how to use random sampling theorem
to estimate the number of untagged \emph{coding} ``U'' pulses from
the number of untagged \emph{coding} ``L'' pulses. This is due to
the correlations between corresponding ``L'' and ``U'' pulses. As
discussed before, their photon numbers are not independent
variables. We are applying a restricted sampling where we draw only
one sample from each pair of U and L pulses.

A common imperfection is the inaccuracy of beam splitting ratio $q$.
One can calibrate the value of $q$, but only with a finite
resolution. In the security analysis, one should pick the most
conservative value of $q$ within the calibrated range. That is, the
value of $q$ that suggests the lowest key generation rate. Similar
strategy should be applied to the inaccuracy of internal
transmittance $\lambda$.

\section{Efficient Passive Estimate on Untrusted Source}

In the above analysis, only half pulses (coding pulses) are used to
generate the secure key. Note that we can also use the measurement
result of coding ``L'' pulses to estimate the number of untagged
sampling ``U'' pulses as there is no physical difference between
sampling pulses and coding pulses. Note that Alice has the
knowledge of the number of untagged coding ``L'' pulses. We have the
following statement:

\begin{cor}\label{Col:2}
Consider $k$ pulses sent from an unknown and untrusted source to
Alice, where $k$ is a large positive integer. Alice randomly assigns
each input pulse as either a sampling pulse or a coding pulse with
equal probabilities. Define variables $V_\mathrm{c}^\mathrm{L}$ and
$V_\mathrm{s}^\mathrm{U}$ as the number of untagged coding L pulses
and the number of untagged sampling U pulses, respectively. Here U
pulses are defined as pulses sent to the Encoder in Figure
\ref{FIG:UL_Pulse}, and L pulses are defined as pulses sent to the
Intensity Monitor in Figure \ref{FIG:UL_Pulse}. Alice can conclude
that $V_\mathrm{s}^\mathrm{U}> V_\mathrm{c}^\mathrm{L}-\epsilon_2 k$
with confidence level $\tau_2 \ge 1-e^{-k\epsilon_2^2/2}$. Here
$\epsilon_2$ is a small positive real number chosen by Alice and Bob.
\end{cor}

A natural question is: Since Alice has the knowledge about both
$V_\mathrm{s}^\mathrm{L}$ and $V_\mathrm{c}^\mathrm{L}$, how can she
estimate the number of total untagged U pulses, $V^\mathrm{U}
(=V_\mathrm{s}^\mathrm{U}+V_\mathrm{c}^\mathrm{U})$?

Combining all untagged U bits is not entirely trivial. Consider that
the untrusted source generates $k$ pulses. Each of them is divided
into 2 pulses. Therefore Alice and Bob have $2k$ pulses to analyze.
However, these $2k$ pulses are \emph{not} independent because the
beam splitter clearly creates correlations between the corresponding
L pulse and U pulse. A na\"{i}ve application of the random sampling
theorem, ignoring the correlation between U pulses and L pulses, may
lead to security loophole.

\begin{lem}

Consider $k$ pulses sent from an unknown and untrusted source to
Alice. Alice randomly assigns each input pulse as either a sampling
pulse or a coding pulse with equal probabilities. Each input pulse
is split into a U pulse and an L pulse (see Figure
\ref{FIG:UL_Pulse} for visualization). The probability that
$V^\mathrm{U}\le V_\mathrm{s}^\mathrm{L}+V_\mathrm{c}^\mathrm{L}-
\epsilon_1 k - \epsilon_2 k$ satisfies:

\begin{equation}\label{EQ:Lemma2}
    P(V^\mathrm{U}\le V_\mathrm{s}^\mathrm{L}+V_\mathrm{c}^\mathrm{L}-
(\epsilon_1 + \epsilon_2) k)\le \exp(\frac{-k\epsilon_1^2}{2}) +
\exp(\frac{-k\epsilon_2^2}{2}).
\end{equation}

\end{lem}

\begin{proof}
See \ref{App:Confidence_Cross}.
\end{proof}

In real experiment, it is convenient to count \emph{all} the
untagged L pulses, defined as variable
$V^\mathrm{L}(=V_\mathrm{s}^\mathrm{L}+V_\mathrm{c}^\mathrm{L})$.
Can we estimate $V^\mathrm{U}$ directly from $V^\mathrm{L}$?

\begin{prop}
Consider $k$ pulses sent from an unknown and untrusted source to
Alice. Alice randomly assigns each input pulse as either a sampling
pulse or a coding pulse with equal probabilities. The probability
that $V^\mathrm{U} \le V^\mathrm{L} - \epsilon k$ satisfies:
\begin{equation}\label{EQ:Prop1}
    P(V^\mathrm{U} \le V^\mathrm{L} - \epsilon k) \le
    2\exp(\frac{-k\epsilon^2}{4})
\end{equation}
That is, Alice can conclude that $V^\mathrm{U} > V^\mathrm{L} -
\epsilon k$ with confidence level
\begin{equation}\label{EQ:Prop1_ConfidenceLevel}
    \tau > 1-2\exp(\frac{-k\epsilon^2}{4})
\end{equation}

\end{prop}

\begin{proof}
This is a natural conclusion from Lemma 2. Note that
$V^\mathrm{L}=V_\mathrm{s}^\mathrm{L}+V_\mathrm{c}^\mathrm{L}$. If
Alice chooses $\epsilon_1=\epsilon_2=\epsilon/2$, Equation
\eqref{EQ:Lemma2} reduces to Equation \eqref{EQ:Prop1}.
\end{proof}

Once the number of untagged bits that are sent to Bob is estimated,
the final key generation rate can be calculated
\cite{Security:UntrustedSource}.

\section{Numerical Simulation}\label{Sec:NumericalSimulation}

We performed numerical simulation to test the efficiencies of the active
and passive estimates. Here, we define the key generation rate as
secure key bits per pulse sent by the \emph{source}, which may be
controlled by an eavesdropper. This is different from the definition
used in \cite{Security:UntrustedSource}, where the key generation
rate is defined as secure key bit per pulse sent by \emph{Alice}.
Note that, in the passive scheme, \emph{all} the pulses sent by the
source are sent from Alice to Bob, while in active scheme, only
\emph{half} of the pulses sent by the source are sent from Alice to
Bob. Therefore, for the same set-up, we can expect the key
generation rate suggested by the passive scheme to be roughly twice
as high as that by the active scheme. However, the equivalent input
photon number in the passive scheme is lower than that of the active
scheme, which introduces a competing factor. The comparison between
passive and active estimates is discussed in following sections.

\subsection{Simulation Techniques}

The simulation technique in this paper is similar to that presented in
\cite{Security:UntrustedSource} with a few improvements. Here we briefly reiterate it:
First, we simulate the experimental outputs based on the parameters
reported by \cite{GYS}, which are shown in Table
\ref{Tab:simulation_parameter}. At this stage, we assume that the
source is Poissonian with an average output photon number $M$. For a QKD setup with channel transmittance $\eta(=e^{-\alpha l}$,
where $\alpha$ is the fiber loss coefficient, and $l$ is the fiber length
between Alice and Bob), Bob's quantum detection efficiency $\eta_\text{Bob}$,
detector intrinsic error rate $e_\text{det}$ and background rate
$Y_0$, the gain \cite{Gain} and the QBER of the signals are expected to be
\cite{Decoy:Practical}
\begin{equation}\label{eq:qmuemu}
    \begin{aligned}
    Q_e &= Y_0+1-\exp(-\eta\eta_\text{Bob}M\lambda),\\
    E_e &=
\frac{e_0Y_0+e_\text{det}[1-\exp(-\eta\eta_\text{Bob}M\lambda)]}{Q_e},
    \end{aligned}
\end{equation} respectively. Here $Q_e$ and $E_e$ refer to the experimentally measured overall properties rather than the properties of the untagged bits.  Second, we calculate the secure key generation rate. The general expression of secure key generation rate per pulse sent by Alice is given by \cite{GLLP, Decoy:LoPRL}
\begin{equation}\label{eq:GeneralR}
R\ge \frac{1}{2}[-Q_ef(E_e)H_2(E_e)+Q_1(1-H_2(e_1))],
\end{equation}
where $f(\ge 1)$ is the bi-directional error correction inefficiency ($f=1$ iff the error correction procedure achieves the Shannon limit), $H_2$ is the binary Shannon entropy, $Q_1$ is the gain of the single photon state in untagged bits, and $e_1$ is the QBER of the single photon state in untagged bits. $Q_e$ and $E_e$ can be experimentally measured. Here, we use Equations \eqref{eq:qmuemu} to simulate the experimental outputs.

$Q_1$ and $e_1$ need to be estimated. Here, we use the method described in \ref{App:Security}. The key assumption for decoy state QKD with an untrusted source is that $Y_{m,n}$ is identical for different states, and so is $e_{m,n}$ \cite{Security:UntrustedSource}. Here $Y_{m,n}$ is the conditional
probability that Bob's detectors click given that this bit enters
Alice's lab with photon number $m$ and emits from Alice's lab with
photon number $n$, and $e_{m,n}$ is the QBER of bits with $m$ input
photons and $n$ output photons.

At the second stage, we do not
make any assumption about the source. That is, Alice and Bob have to
characterize the source from the experimental output. Note that we need to set the values of $\lambda$ and $\delta$ (recall that all untagged bits have input photon numbers
$m\in[(1-\delta)M,(1+\delta)M]$, where $\delta$ is a small positive real
number, $M$ is a large positive integer, and both $\delta$ and $M$
are chosen by Alice and Bob). It is preferable to set $\lambda$ and $\delta$ to the values that yield the highest final key generation rate. We optimize the values of $\lambda$ and $\delta$ numerically by exhaustive search. Moreover, in the simulation of decoy state QKD with a finite data size, we also need to optimize the portion of each state.

As a clarification, our security analysis does \emph{not} require
any additional assumptions of the source to analyze
\emph{experimental} outputs.

An important improvement is that the value of $\delta$
 is optimized at all distances in the following simulations, while
$\delta$ is set to be constant in \cite{Security:UntrustedSource}.
This is because for different channel losses, the optimal value of
$\delta$ can vary. Moreover, several important practical factors are
considered, including the unique characteristic of plug \& play
structure, intensity monitor imperfections, and finite data size.

\begin{figure}\center
  \includegraphics[width=9cm]{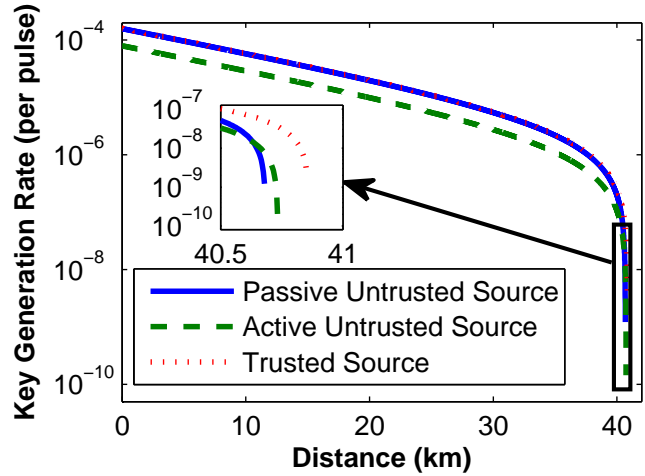}\\
  \caption{Simulation result of GLLP \cite{GLLP} protocol with infinite
  data size, symmetric beam splitter, perfect intensity monitor, and uni-directional structure.
  We assume that the source is Poissonian centered at $M=10^6$
  photons per pulse, and the beam splitting ratio $q=0.5$.
  Citing experimental parameters from Table \ref{Tab:simulation_parameter}. We calculated the ratio
  of the key generation rate with an untrusted source over that with a trusted
source. For the passive estimate scheme, the ratios are 98.4\%, 98.1\%,
and 79.8\% at 1 km, 20 km, and 40 km, respectively. For the active
estimate scheme, the ratios are 49.4\%, 49.3\%, and 42.8\% at 1 km,
20 km, and 40 km, respectively.}\label{FIG:GLLP_GYS_PerfectIM}
\end{figure}

For ease of calculation, similar to in
\cite{Security:UntrustedSource}, we approximate the Poisson
distribution as a Gaussian distribution centered at $M$ with
variance $\sigma^2=M$. This is an excellent approximation because
$M$ is very large ($10^3$ or larger) in all the simulations
presented below.

There are various types of imperfections and errors. We will consider them one by one in the following sections. In Section \ref{se:biasBS}, we consider the asymmetry of the beam splitter. In Section \ref{se:pnp}, we consider the source attenuation introduced by the bi-directional scheme. In Section \ref{Sec:SimImperfectIM}, we consider the inefficiency and the inaccuracy of the intensity monitor. In Section \ref{se:finite}, we consider the statistical fluctuation due to a finite data size.

\subsection{Infinite Data Size with Perfect Intensity Monitor}
In the asymptotic case, Alice sends infinitely many bits to Bob (i.e.,
$k\rightarrow \infty$). Therefore we can set $\epsilon\rightarrow0$
while still having $\tau\rightarrow 1$.

We assume that the intensity monitor is efficient and noiseless.
Similarly to the case in \cite{Security:UntrustedSource}, we set $M=10^6$.
Moreover, we set $q=0.5$ as 50/50 beam splitter is widely used in
many applications.

\begin{table}[!b]
  \center
  \caption{Simulation Parameter from GYS \cite{GYS}.}\label{Tab:simulation_parameter}
    \begin{tabular}{c c c c}
      \hline
      $\eta_\mathrm{det}$ & $\alpha$ & $Y_0$ & $e_\mathrm{det}$ \\
      4.5\% & 0.21dB/km & $1.7\times 10^{-6}$ & 3.3\% \\
      \hline
    \end{tabular}
\end{table}

The simulation results of the GLLP protocol \cite{GLLP}, Weak+Vacuum
decoy state protocol \cite{Decoy:Practical}, and One-Decoy protocol
\cite{Decoy:Practical} are shown in Figure
\ref{FIG:GLLP_GYS_PerfectIM}, Figure \ref{FIG:WV_GYS_PerfectIM}, and
Figure \ref{FIG:OneDecoy_GYS_PerfectIM}, respectively. We can see
that the key generation rate of the passive estimate scheme on an untrusted
source is very close to that on  a trusted source, while the key
generation rate of the active estimate scheme is roughly $1/2$ of that
of the passive scheme. This is expected because in the active scheme, only
half of the pulses generated by the source are sent to Bob, whereas
in the passive scheme, all the pulses generated by the source are sent
to Bob. Note that, in the asymptotic case, the efficiency of the active
estimate scheme can be doubled by sending most pulses
(asymptotically all the pulses) to Bob. In this case, there are
still infinitely many pulses sent to the Intensity Monitor.

\begin{figure}\center
  \includegraphics[width=9cm]{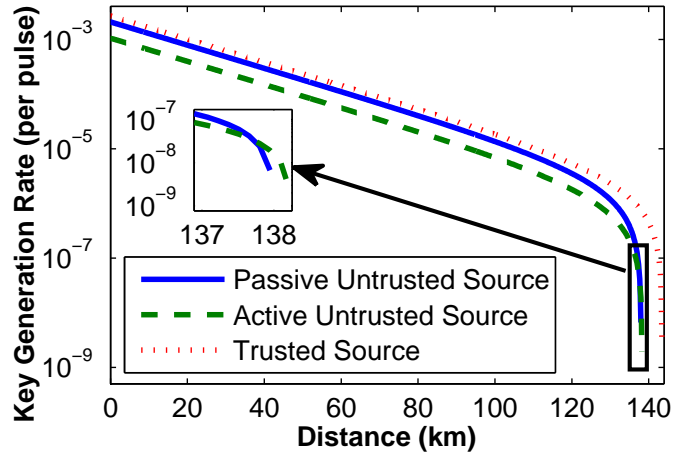}\\
  \caption{Simulation result of Weak+Vacuum \cite{Decoy:Practical} protocol with infinite
  data size, symmetric beam splitter, perfect intensity monitor, and uni-directional structure.
   We assume that the source is Poissonian centered at $M=10^6$
  photons per pulse, and the beam splitting ratio $q=0.5$.
  Citing experimental parameters from Table \ref{Tab:simulation_parameter}. We calculated the ratio
  of the key generation rate with an untrusted source over that with a trusted
source. For the passive estimate scheme, the ratios are 77.7\%, 77.1\%,
and 73.8\% at 1 km, 50 km, and 100 km, respectively. For the active
estimate scheme, the ratios are 39.2\%, 39.0\%, and 37.4\% at 1 km,
50 km, and 100 km, respectively.}\label{FIG:WV_GYS_PerfectIM}
\end{figure}

For ease of discussion, in passive estimate scheme, we define
untagged bits as bits with input photon number
$m_\mathrm{p}\in[(1-\delta_\mathrm{p})M_\mathrm{p},(1+\delta_\mathrm{p})M_\mathrm{p}]$,
while in active estimate scheme, we define untagged bits as bits
with input photon number
$m_\mathrm{a}\in[(1-\delta_\mathrm{a})M_\mathrm{a},(1+\delta_\mathrm{a})M_\mathrm{a}]$.
Here $\delta_\mathrm{p}$ and $\delta_\mathrm{a}$ are small positive real
numbers chosen by Alice and Bob, and $M_\mathrm{p}$ and
$M_\mathrm{a}$ are large positive integers chosen by Alice and Bob.
In the passive estimate scheme, we define the maximum possible tagged
ratio as $\Delta_\mathrm{p}$. In active estimate scheme, we define
the maximum possible tagged ratio as $\Delta_\mathrm{a}$. Here the
tagged ratio is defined as the ratio of the number of tagged bits
over the number of all the bits sent to Bob.

By magnifying the tails at long distances (shown in the insets of
Figures
\ref{FIG:GLLP_GYS_PerfectIM} -- \ref{FIG:OneDecoy_GYS_PerfectIM}), we
can see that the active schemes suggest higher key generation rate than
the passive schemes do in all three protocols. This behavior is related
to the following fact: In the passive estimate scheme, the
equivalent input photon number is lower than that of the active
estimate scheme. This is because the input photon number is defined
as the photons counted by the intensity monitor, and only a portion
of an input pulse is sent to the intensity monitor in the passive
scheme. Compared to the active scheme, lower input photon number in the
passive scheme leads to a larger coefficient of variation of measured
input photon number distribution, assuming the source is Poissonian.
Therefore, for the same source, if one set
$\delta_\mathrm{p}=\delta_\mathrm{a}$, $\Delta_\mathrm{p}$ will be
greater than $\Delta_\mathrm{a}$ \footnote{The values of $\delta$ in the
passive estimate and the  active estimate schemes are optimized separately in our
simulation. The optimal value of $\delta_\text{p}$ usually deviates
from the optimal value of $\delta_\text{a}$ with the same
experimental parameters. Here we cite
``$\delta_\text{p}=\delta_\text{a}$'' just to illustrate an
intuitive understanding of the phenomena shown in the insets of Figures 
\ref{FIG:GLLP_GYS_PerfectIM} -- \ref{FIG:OneDecoy_GYS_PerfectIM}.}.
Increasing the coefficient of variation of the measured input photon number
distribution will in general deteriorate the efficiency of the estimate
for QKD with untrusted sources. Take two extreme cases for example:
If the coefficient of variation is very large, which means the input
photon number distribution is almost a uniform distribution, then
the estimate efficiency will be very poor because either $\delta$ or
$\Delta$ (or both) will be very large. If the coefficient of
variation is very small, which means the input photon number
distribution is almost a delta-function, then the estimate
efficiency will be very good because both $\delta$ and $\Delta$ can
be very small.

\begin{figure}\center
  \includegraphics[width=9cm]{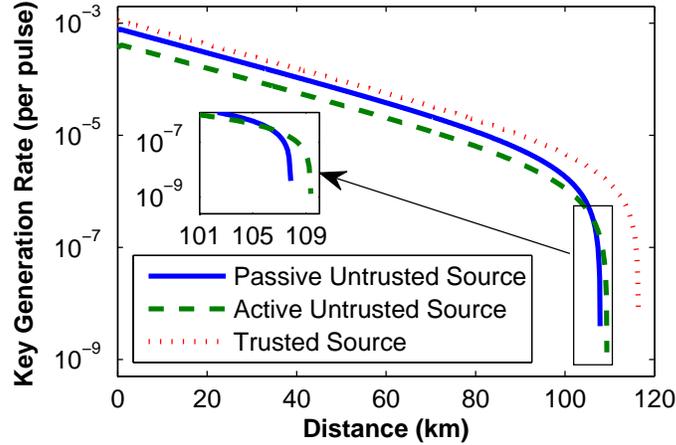}\\
  \caption{Simulation result of One-decoy \cite{Decoy:Practical}
  protocol with infinite
  data size, symmetric beam splitter,
  perfect intensity monitor, and uni-directional structure.
   We assume that the source is Poissonian centered at $M=10^6$
  photons per pulse, and the beam splitting ratio $q=0.5$.
  Citing experimental parameters from Table \ref{Tab:simulation_parameter}. We calculated the ratio
  of the key generation rate with an untrusted source over that with a trusted
source. For the passive estimate scheme, the ratios are 71.5\%, 68.6\%,
and 39.5\% at 1 km, 50 km, and 100 km, respectively. For the active
estimate scheme, the ratios are 38.0\%, 36.7\%, and 24.4\% at 1 km,
50 km, and 100 km, respectively.}\label{FIG:OneDecoy_GYS_PerfectIM}
\end{figure}

The estimate of the gain of untagged bits is very sensitive to the
value of $\Delta$, especially when the experimental measured overall
gain is small (i.e., when the distance is long, which corresponds to
the tails of Figures
\ref{FIG:GLLP_GYS_PerfectIM} -- \ref{FIG:OneDecoy_GYS_PerfectIM}). The
estimate of untagged bits' gain is discussed in Section III of
\cite{Security:UntrustedSource}. Here we briefly recapitulate the
main idea: Alice cannot in practice perform a quantum non-demolition
measurement on the photon numbers of input pulses. Therefore, Alice
and Bob do not know which bits are tagged and which are untagged,
though they can estimate the minimum number of untagged bits.
Without knowing which bits are untagged, Alice and Bob cannot
measure the exact gain $Q$ of untagged bits. Alice and Bob can only
experimentally measure the overall gain $Q_e$, which contains
contributions from both tagged bits and untagged bits.

Alice and Bob can still estimate the upper and lower bounds of $Q$.
They can first estimate the maximum tagged ratio $\Delta$. This
estimate can be obtained either actively as proposed in
\cite{Security:UntrustedSource}, or passively as discussed in this
paper. Alice and Bob can then estimate the upper and lower bounds of
$Q$ as follows \cite{Security:UntrustedSource}:
\begin{equation}\label{EQ:Q_Bound}
    \begin{aligned}
        \overline{Q} &= \frac{Q_e}{1-\Delta-\epsilon},\\
        \underline{Q} &=
        \max\left(0,\frac{Q_e-\Delta-\epsilon}{1-\Delta-\epsilon}\right);
    \end{aligned}
\end{equation}

\begin{figure}\center
\subfigure[$q=0.5$]{
\includegraphics[width=5cm]{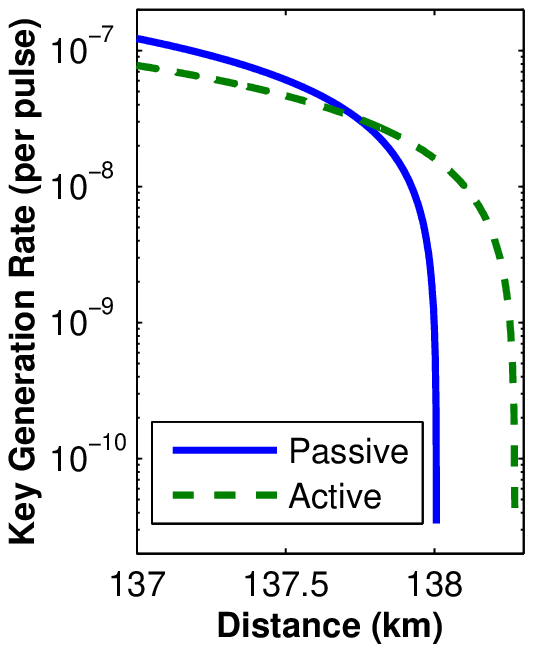}}
\subfigure[$q=0.01$]{
\includegraphics[width=4.85cm]{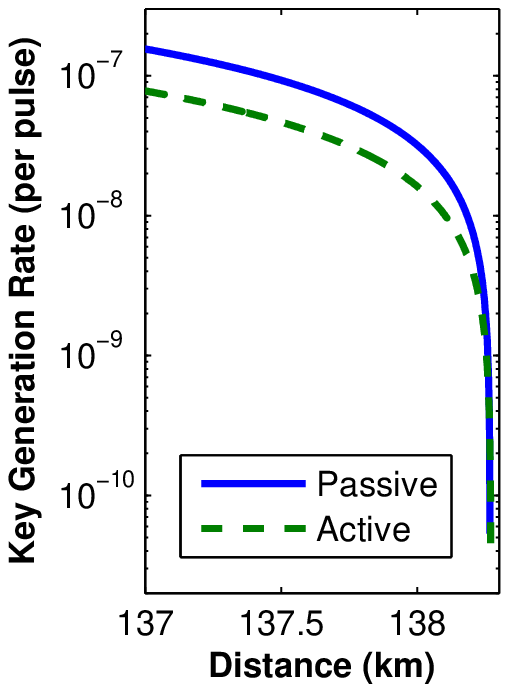}}\\
  \caption{Simulation results for Weak+Vacuum protocol
  \cite{Decoy:Practical} with different beam splitters
  for passive estimate.
   We assume that the data size is infinite, the intensity
  monitor is perfect, the source is Poissonian centered at $M=10^6$
  photons per pulse,
  and the system is in uni-directional structure.
  Citing experimental parameters from Table \ref{Tab:simulation_parameter}.
  The results are focused at
  the maximum transmission distance to illustrate the
  improvement of passive
  estimate by using a biased beam splitter that sends more
  photons into the intensity
  monitor. This is equivalent to increasing input photon numbers in passive
  scheme.}\label{FIG:Different_q}
\end{figure}

$\underline{Q}$ is very sensitive to $\Delta$ when $Q_e$ is small.
Therefore, when the distance is long (which corresponds the 
tails of Figures
\ref{FIG:GLLP_GYS_PerfectIM} -- \ref{FIG:OneDecoy_GYS_PerfectIM}),
$Q_e$ becomes very small, and $\underline{Q}$ will then be very
sensitive to $\Delta$. Since $\Delta_\mathrm{p}>\Delta_\mathrm{a}$,
the passive estimate becomes less efficient than the active estimate
in this case.

On the other hand, in short distances, $Q_e$ is significantly
greater than $\Delta_\mathrm{p}$ and $\Delta_\mathrm{a}$, therefore
the difference between $\Delta_\mathrm{p}$ and $\Delta_\mathrm{a}$
makes a negligible contribution to the performance difference between
the passive and active estimates. At short distances, it is the
following fact that dominates the performance difference between
these two schemes: The passive estimate scheme can send Bob twice as many
pulses as the active estimate scheme can.

One can increase $\delta$ to decrease $\Delta_\mathrm{p}$. That is,
if one intends to ensure that $\Delta_\mathrm{p}=\Delta_\mathrm{a}$,
one has to set $\delta_\mathrm{p}>\delta_\mathrm{a}$. However,
increasing $\delta$ also has negative effect on the key generation
rate. This is discussed in Section III \& IV of
\cite{Security:UntrustedSource}.

In brief, lower input photon number is the reason why the passive
estimate suggests lower key generation rate than the active estimate
does around maximum transmission distances in all of the three
simulated protocols. This will be confirmed in the simulation
presented in Section \ref{se:biasBS} --
\ref{se:finite}.

\begin{figure}\center
  \includegraphics[width=9cm]{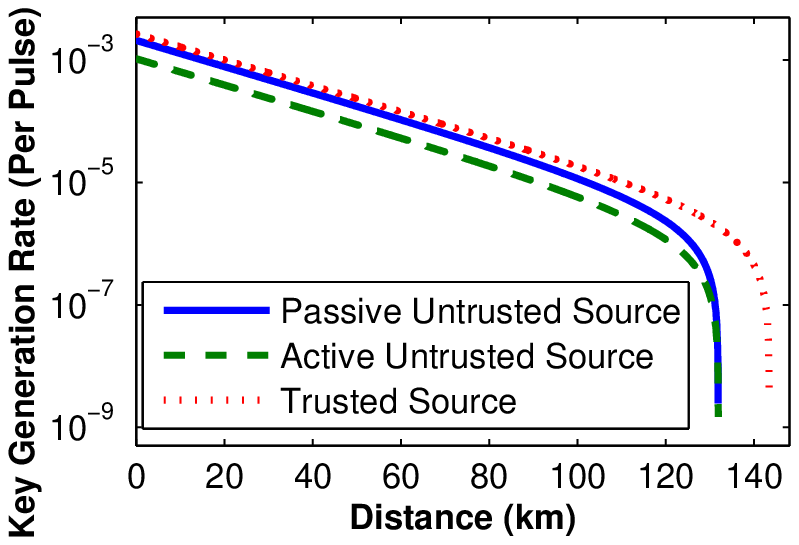}\\
  \caption{Simulation result of Weak+Vacuum \cite{Decoy:Practical}
  protocolwith infinite
  data size, asymmetric beam splitter, perfect intensity monitor,
  and \emph{bi-directional structure}.
   We assume
  that the source in Bob's lab is Poissonian centered at $M_\mathrm{B}=10^6$
  photons per pulse, and the beam splitting ratio $q=0.01$.
  Citing experimental parameters from Table \ref{Tab:simulation_parameter}.
  We calculated the ratio
  of the key generation rate with an untrusted source over that with a trusted
source. For passive estimate scheme, the ratios are 78.5\%, 75.0\%,
and 63.0\% at 1 km, 50 km, and 100 km, respectively. For active
estimate scheme, the ratios are 39.2\%, 37.5\%, and 31.5\% at 1 km,
50 km, and 100 km, respectively. Comparing with Figure
\ref{FIG:WV_GYS_PerfectIM}, we can see that the bi-directional
nature of Plug \& Play set-up reduced the efficiencies of both
active and passive estimates on an untrusted
source.}\label{FIG:WV_GYS_PerfectIM_PnP}
\end{figure}

\subsection{Biased Beam Splitter}\label{se:biasBS}

A natural measure to improve the efficiency of the passive estimate
is to increase input photon number. Note that in the passive estimate,
as discussed in Section \ref{Sec:ActiveToPassive}, input photon
numbers are the photon numbers counted by the intensity monitor.
Therefore, it can improve the passive estimate's efficiency to send
more photons to the intensity monitor (i.e., setting $q$ smaller).

To test this postulate, we performed another simulation to compare
the performance of the passive estimate with different values of $q$.
Similar to the above subsection, we assume that the intensity
monitor is efficient and noiseless, and data size is infinite.
Therefore $\epsilon=0$. We set $M=10^6$ at the \emph{source}.

The simulation results are shown in Figure \ref{FIG:Different_q}. We can clearly
see that by setting $q$ to a smaller value (1\%), the key generation
rate of the passive estimate scheme is improved around the maximum
transmission distance.

Intuitively, one can improve the efficiency of the active scheme by
sending most pulses to Bob. One can refer to the discussion in 
\ref{App:Confidence_Active} below Equation
\eqref{Eq:SamplingCodingDeviation_General} as a starting point.
Detailed discussion of optimizing the efficiency of the active estimate
scheme is beyond the scope of the current paper and is subject to
further investigation.

\begin{figure}\center
  \includegraphics[width=9cm]{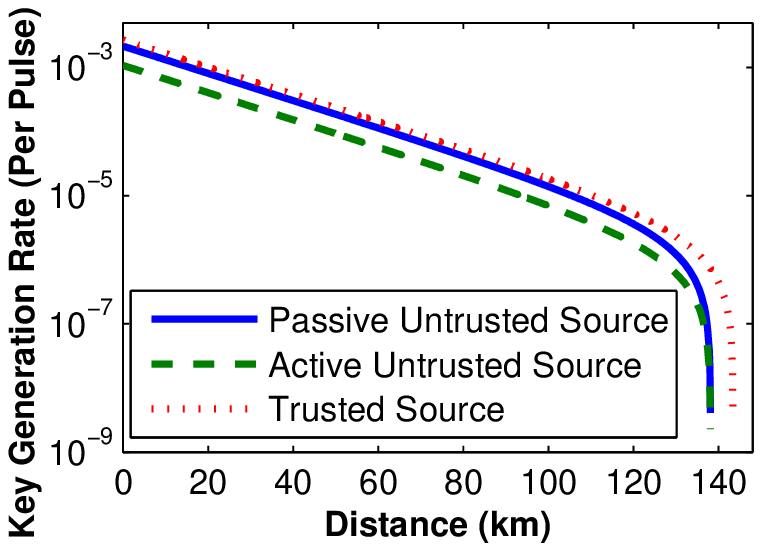}\\
  \caption{Simulation result of Weak+Vacuum \cite{Decoy:Practical} protocol
  with infinite
  data size, asymmetric beam splitter, perfect intensity monitor,
   bi-directional structure, and \emph{a bright light source}.
   We assume
  that the source in Bob's lab is Poissonian centered at $M_\mathrm{B}=\mathbf{10^8}$
  photons per pulse, and the beam splitting ratio $q=0.01$.
  Citing experimental parameters from Table \ref{Tab:simulation_parameter}.
  We calculated the ratio
  of the key generation rate with an untrusted source over that with a trusted
source. For the passive estimate scheme, the ratios are 80.3\%, 79.6\%,
and 75.8\% at 1 km, 50 km, and 100 km, respectively. For the active
estimate scheme, the ratios are 40.1\%, 39.8\%, and 37.9\% at 1 km,
50 km, and 100 km, respectively. Comparing with Figure
\ref{FIG:WV_GYS_PerfectIM_PnP}, we can see that the estimate
efficiencies for both the active and passive schemes are improved by
using a brighter source. }\label{FIG:WV_GYS_PerfectIM_PnP_LargeN}
\end{figure}

\subsection{Plug \& Play Setup}\label{se:pnp}
In the Plug \& Play QKD scheme, the source is located in Bob's lab.
Bright pulses sent by Bob will suffer the whole channel loss before
entering Alice's lab. Therefore, in the Plug \& Play set-up, Alice's
average input photon number is dependent on the channel loss between
Alice and Bob. If the average photon number per pulse at the source
in Bob's lab, $M_\mathrm{B}$, is constant, the average input photon
number per pulse in Alice's lab, $M$, decreases as the channel loss
increases.

Similar to in the above subsection, we assume that the intensity
monitor is efficient and noiseless, and data size is infinite.
Therefore $\epsilon=0$. We set $M_\mathrm{B}=10^6$ at the source in
Bob's lab. We set $q=1\%$ to improve the passive estimate
efficiency.

We clarify that ``distance'' in all the simulations of
bi-directional QKD set-up refers to a one-way distance between Alice
and Bob, \emph{not} a round-trip distance.

The simulation results of Weak+Vacuum protocol
\cite{Decoy:Practical} are shown in Figure
\ref{FIG:WV_GYS_PerfectIM_PnP}. We can see that the bi-directional
nature plug \& play structure clearly deteriorates the performance
at long distances for which the input photon number on Alice's side is
largely reduced. This affects both the passive and active estimates.

A natural measure to improve the performance of the Plug \& Play setup
is to use a brighter source. By setting $M_\mathrm{B}=10^8$ at the
source in Bob's lab, the performances for both passive and active
estimates are improved substantially as shown in Figure
\ref{FIG:WV_GYS_PerfectIM_PnP_LargeN}. Note that subnanosecond
pulses with $\sim10^8$ photons per pulse can be routinely generated
with directly modulated laser diodes.

\begin{figure}\center
  \includegraphics[width=9cm]{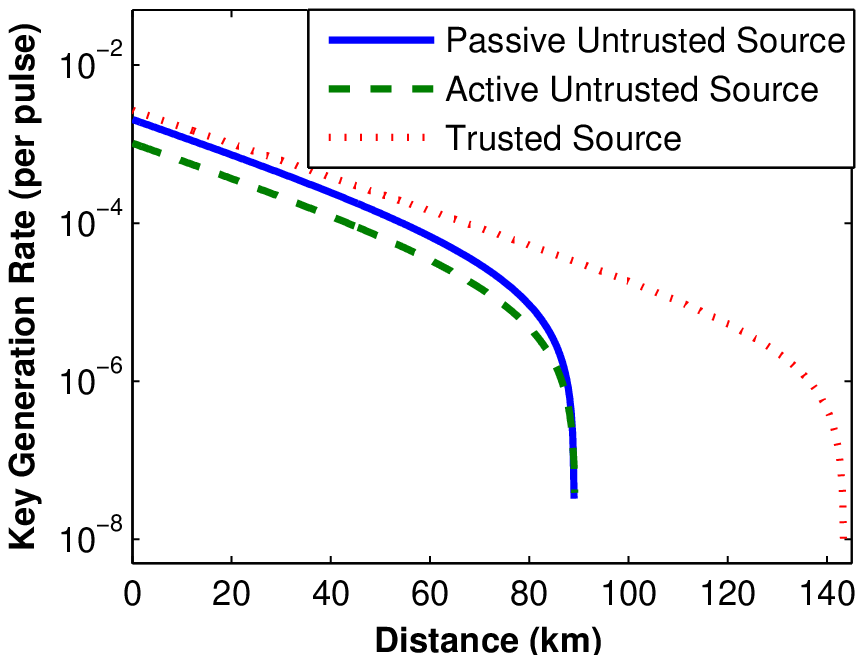}\\
  \caption{Simulation result of Weak+Vacuum \cite{Decoy:Practical}
  protocol with infinite
  data size, asymmetric beam splitter, \emph{imperfect intensity monitor},
  and bi-directional structure.
    We assume
  that the intensity
  monitor efficiency $\eta_\mathrm{IM}=0.7$, the intensity
  monitor noise $\sigma_\mathrm{IM}=10^5$, the intensity monitor
  conservative interval $\varsigma=6\times10^5$, the source in Bob's lab is Poissonian
  centered at $M_\mathrm{B}=10^8$
  photons per pulse, and the beam splitting ratio $q=0.01$.
  Citing experimental parameters from Table \ref{Tab:simulation_parameter}.
  Comparing with Figure \ref{FIG:WV_GYS_PerfectIM_PnP}, we can see that
  the imperfections of the intensity monitor substantially
  reduce the efficiencies of both active and passive estimates.}\label{FIG:WV_GYS_ImPerfectIM_PnP}
\end{figure}

\subsection{Imperfections of the Intensity
Monitor}\label{Sec:SimImperfectIM}

There are two major imperfections of the intensity monitor:
inefficiency and noise. These imperfections are discussed in Section
\ref{Sec:ActiveToPassive}. The inefficiency can be easily modeled as
additional loss in the simulation.

There can be various noise sources, including thermal noise, shot-noise, etc. Here, we consider a simple noise model where a \emph{constant
Gaussian} noise with variance $\sigma_\mathrm{IM}^2$ is assumed.
That is, if $m$ photons enter an efficient but noisy intensity
monitor, the probability that the measured photon number is $m'$
obeys a Gaussian distribution
$$
P_m(m')=\frac{1}{\sigma_\mathrm{IM}\sqrt{2\pi}}\exp[-\frac{(m-m')^2}{2\sigma_\mathrm{IM}^2}].
$$

The measured photon number distribution $P(m')$ has larger variation
than the actual photon number distribution $P(m)$ due to the noise
of the intensity monitor. More concretely, if the actual photon
numbers obeys a Gaussian distribution centered at $M$ with variance
$\sigma^2$, the measured photon numbers also obeys a Gaussian
distribution centered at $M$, but with a variance
$\sigma^2+\sigma_\mathrm{IM}^2$.

As in the previous subsections, we assume that the data size is
infinite. Therefore $\epsilon=0$. We set $M_\mathrm{B}=10^8$ at the
source in Bob's lab. Plug \& Play set-up is assumed. We set $q=1\%$
to improve the passive estimate efficiency. The imperfections of the
intensity monitor are set as follows: the efficiency is set as
$\eta_\mathrm{IM}=0.7$, and the noise is set as
$\sigma_\mathrm{IM}=10^5$ (see experimental parameters in Section
\ref{Sec:SimPKU} and Section \ref{Sec:PreExp}). For ease of
simulation, we assume that the intensity monitor conservative
interval is constant \footnote{The assumption of constant
conservative interval may not precisely describe the inaccuracy of
the intensity monitor in realistic applications. Nonetheless, some
factors, like finite resolution of analog-digital conversion, may 
indeed be constant at different intensity levels. We remark that the noises
of different intensity monitors may vary largely. Detailed
investigation on the intensity monitor noise modeling is beyond the
scope of the current paper.} over different input photon numbers. We set
$\varsigma=6\sigma_\mathrm{IM}=6\times10^5$ to ensure a conservative
estimate.

\begin{figure}[!t]\center
  \includegraphics[width=9cm]{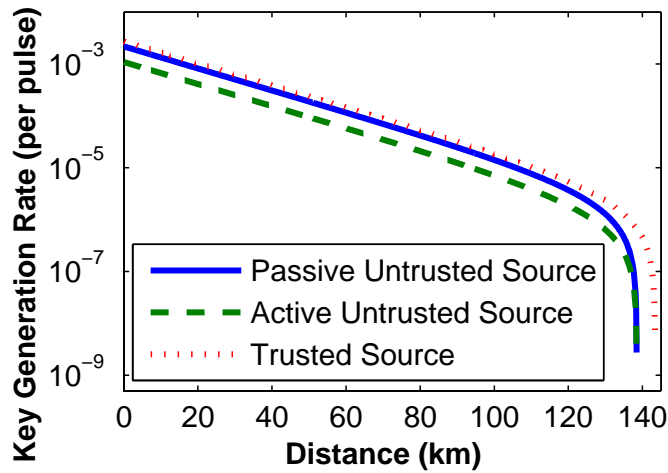}\\
  \caption{Simulation result of Weak+Vacuum \cite{Decoy:Practical}
  protocol with infinite
  data size, asymmetric beam splitter, imperfect intensity monitor,
   bi-directional structure, and a \emph{very bright source}.
    We assume
  that the intensity
  monitor efficiency $\eta_\mathrm{IM}=0.7$, the intensity
  monitor noise $\sigma_\mathrm{IM}=10^5$, the intensity monitor
  conservative interval $\varsigma=6\times10^5$, the source in Bob's lab is Poissonian
  centered at $M_\mathrm{B}=10^{10}$
  photons per pulse, and the beam splitting ratio $q=0.01$.
  Citing experimental parameters from Table \ref{Tab:simulation_parameter}. Comparing
  with Figure \ref{FIG:WV_GYS_ImPerfectIM_PnP}, we can see that using
  a brighter source can effectively improve the efficiencies
  of both passive and
  active estimates. Although it is challenging
to build such bright pulsed laser diodes ($10^{10}$ photons per
pulse with pulse width less than 1 ns) at telecom wavelengths, one
can simply attach a fibre amplifier to the laser diode to generate
very bright pulses. Nonetheless, at such a high intensity level,
non-linear effects in the fibre, like self phase modulation, may be
significant \cite{PassiveEstimate_NonLinear}.
}\label{FIG:WV_GYS_ImPerfectIM_PnP_LargeM}
\end{figure}

The simulation results for Weak+Vacuum protocol
\cite{Decoy:Practical} are shown in Figure
\ref{FIG:WV_GYS_ImPerfectIM_PnP}. We can see that the detector noise
significantly affects the performance of the Plug \& Play QKD system.
This is because at long distances, the bi-directional nature of the Plug
\& Play set-up reduces the input photon number on Alice's side.
Intensity monitor noise and the conservative interval are assumed as
constants regardless of the input photon number in our simulation.
Therefore they become critical issues when the input photon number is
low. As a result, the key generation rate at long distance is
substantially reduced.

\begin{figure}\center
  \includegraphics[width=9cm]{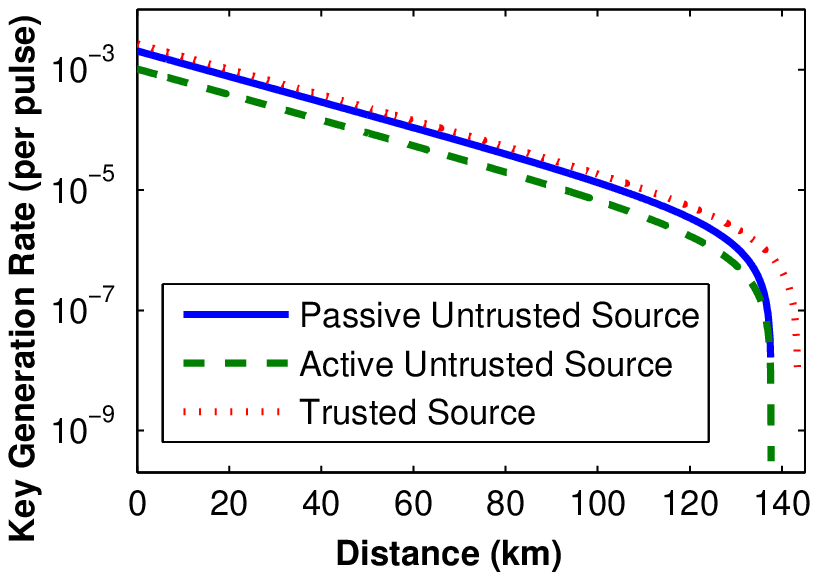}\\
  \caption{Simulation result of Weak+Vacuum \cite{Decoy:Practical}
  protocol with infinite
  data size, asymmetric beam splitter, imperfect intensity monitor,
  and \emph{uni-directional structure}.
    We assume
  that the intensity
  monitor efficiency $\eta_\mathrm{IM}=0.7$, the intensity
  monitor noise $\sigma_\mathrm{IM}=10^5$, the intensity monitor
  conservative interval $\varsigma=6\times10^5$, the source is Poissonian
  centered at $M=10^8$
  photons per pulse, and the beam splitting ratio $q=0.01$.
  Citing experimental parameters from Table \ref{Tab:simulation_parameter}.
  Comparing
  with Figure \ref{FIG:WV_GYS_ImPerfectIM_PnP}, we can see that uni-directional
  structure can effectively improve the efficiencies of both passive and
  active estimates.}\label{FIG:WV_GYS_ImPerfectIM}
\end{figure}

The above postulate is confirmed by the simulations shown in Figure
\ref{FIG:WV_GYS_ImPerfectIM_PnP_LargeM} and Figure
\ref{FIG:WV_GYS_ImPerfectIM}. In Figure
\ref{FIG:WV_GYS_ImPerfectIM_PnP_LargeM}, we assume that the source
in Bob's lab is extremely bright (sending out $10^{10}$ photons per
pulse). We can see clearly that when the input photon number on
Alice's side is high, the key generation rate is only affected
slightly by the imperfections of the intensity monitor. Although it
is challenging to build such bright pulsed laser diodes ($10^{10}$
photons per pulse with pulse width less than 1 ns) at telecom
wavelengths, one can simply attach a fibre amplifier to the laser
diode to generate very bright pulses. Nonetheless, at such a high
intensity level, non-linear effects in the fibre, like self phase
modulation, may be significant \cite{PassiveEstimate_NonLinear}.

An alternative solution is to use the uni-directional setting, in which
the photon number per pulse is constantly high on Alice's side. From
Figure \ref{FIG:WV_GYS_ImPerfectIM} we can see that using the
uni-directional setting can also minimize the negative effects
introduced by the imperfections of the intensity monitor. Nonetheless, if one adopts the uni-directional QKD scheme, one will lose the unique advantages of bi-directional QKD scheme, like the intrinsic stability against the polarization dispersion and the phase drift. Note that adopting the uni-directional scheme does not mean the coherent state assumption is valid. Indeed, even if Alice possesses the source, the source may not be Poissonian and Alice may not have a full characterization of the source without real-time monitoring.

\subsection{Finite Data Size}\label{se:finite}

Real experiments are performed within a limited time, during which the
source can only generate a finite number of pulses. To be consistent
with previous analysis, we assume that the source generates $k$
pulses in an experiment. Reducing the data size from infinite to
finite has two consequences: First, if the confidence level $\tau$
as defined in Equation \eqref{EQ:Prop1_ConfidenceLevel} (for passive
estimate) or in Equation \eqref{EQ:Lemma1_ConfidenceLevel} (for
active estimate) is expected to be close to 1, $\epsilon$ has to be
positive. More concretely, for a fixed $k$, if the estimate on the
untrusted source is expected to have confidence level no less than
$\tau$, one has to pick $\epsilon$ as
$$
\epsilon_\mathrm{p} = \sqrt{-\frac{4\ln(\frac{1-\tau}{2})}{k}}
$$
in the passive estimate scheme, or
$$
\epsilon_\mathrm{a} = \sqrt{-\frac{2\ln(1-\tau)}{k}}
$$
in the active estimate scheme. Second, in decoy state protocols
\cite{Decoy:Practical}, the statistical fluctuations of experimental
outputs have to be considered. The technique to analyze the statistical
fluctuation in decoy state protocols for numerical simulation is
discussed in \cite{Decoy:Practical,Decoy:WangPRA,Decoy:ZhaoISIT}.

\begin{figure}\center
  \includegraphics[width=9cm]{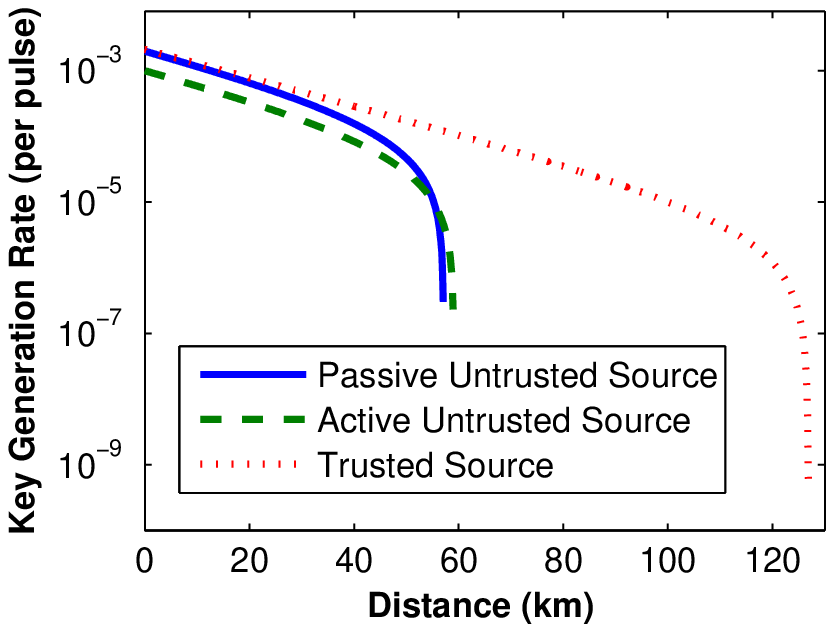}\\
  \caption{Simulation results of the Weak+Vacuum \cite{Decoy:Practical}
  protocol with \emph{finite
  data size}, asymmetric beam splitter, imperfect intensity monitor,
  and bi-directional structure. We assume
  that the data size is $\mathbf{10^{12}}$, the intensity
  monitor efficiency $\eta_\mathrm{IM}=0.7$, the intensity
  monitor noise $\sigma_\mathrm{IM}=10^5$, the intensity monitor
  conservative interval $\varsigma=6\times10^5$, the source in Bob's lab is Poissonian
  centered at $M_\mathrm{B}=10^8$
  photons per pulse, the beam splitting ratio $q=0.01$.
  Confidence level is set as
  $\tau\ge1-10^{-10}$. 6 standard deviations are considered in the statistical
  fluctuation.
  Citing experimental parameters from Table \ref{Tab:simulation_parameter}. Comparing with
  Figure \ref{FIG:WV_GYS_ImPerfectIM_PnP}, we can see that
  finite data size reduces efficiencies of both
  active and passive estimates.}\label{FIG:WV_GYS_ImPerfectIM_PnP_Finite_LargeDataSize}
\end{figure}

In the simulation presented in Figure
\ref{FIG:WV_GYS_ImPerfectIM_PnP_Finite_LargeDataSize}, we assume
that the data size is $10^{12}$ bits (i.e., the source generates
$10^{12}$ pulses in one experiment). This data size is reasonable
for the optical layer of the QKD system because because reliable
gigahertz QKD implementations have been reported in several recent works
\cite{DPSK:200km,ExpFastQKD:MbpsYamamoto,ExpFastQKD:MbpsShields}.
$10^{12}$ bits can be generated within a few minutes in these
gigahertz QKD systems. We set the confidence level as $\tau \ge
1-10^{-10}$, which suggests $\epsilon_\mathrm{a}=6.79\times10^{-5}$
and $\epsilon_\mathrm{p}=9.74\times10^{-5}$. We consider 6 standard
deviations in the statistical fluctuation analysis of Weak+Vacuum
protocol.

As in the previous subsections, we set $M_\mathrm{B}=10^8$ at the
source in Bob's lab. A Plug \& Play set-up is assumed. We set $q=1\%$
to improve the passive estimate efficiency. The imperfections of the
intensity monitor are set as follows: the efficiency is set as
$\eta_\mathrm{IM}=0.7$, and the noise is set constant as
$\sigma_\mathrm{IM}=10^5$. The intensity monitor conservative
interval is set constant as
$\varsigma=6\sigma_\mathrm{IM}=6\times10^5$.

The simulation results for the Weak+Vacuum protocol
\cite{Decoy:Practical} are shown in Figure
\ref{FIG:WV_GYS_ImPerfectIM_PnP_Finite_LargeDataSize}. We can see
that finite data size clearly reduces the efficiencies of both
active and passive estimates. The aforementioned two consequences of
finite data size contribute to this efficiency reduction: First,
$\epsilon$ is non-zero in this finite data size case. Therefore, the
estimate of the lower bound of untagged bits' gain is worse as
reflected in Equation \eqref{EQ:Q_Bound}. Note that $\epsilon$ has
the same weight as $\Delta$ in Equation \eqref{EQ:Q_Bound}. Second,
the statistical fluctuation for  the Weak+Vacuum protocol becomes
important \cite{Decoy:ZhaoISIT}. Moreover, the tightness of bounds
suggested in Lemma 1, Lemma 2, and Proposition 1 may also affect the
estimate efficiency in finite data size.

\begin{figure}\center
  \includegraphics[width=9cm]{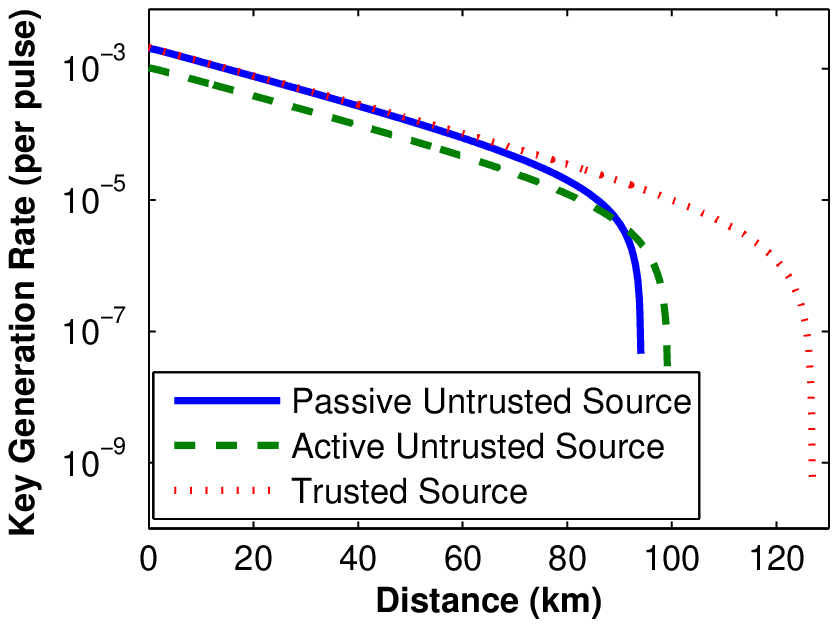}\\
  \caption{Simulation result of Weak+Vacuum \cite{Decoy:Practical}
  protocol with finite
  data size, asymmetric beam splitter, imperfect intensity monitor,
   bi-directional structure, and \emph{very bright source}. We assume
  that the data size is $10^{12}$, the intensity
  monitor efficiency $\eta_\mathrm{IM}=0.7$, the intensity
  monitor noise $\sigma_\mathrm{IM}=10^5$, the intensity monitor
  conservative interval $\varsigma=6\times10^5$, the source in Bob's lab is Poissonian
  centered at $M_\mathrm{B}=\mathbf{10^{10}}$
  photons per pulse, the beam splitting ratio $q=0.01$.
  Confidence level is set as
  $\tau\ge1-10^{-10}$. 6 standard deviations are considered in the statistical
  fluctuation.
  Citing experimental parameters from Table \ref{Tab:simulation_parameter}. Comparing with
  Figure \ref{FIG:WV_GYS_ImPerfectIM_PnP}, we can see that
  using a very bright source can improve efficiencies of both
  active and passive estimates.}\label{FIG:WV_GYS_ImPerfectIM_PnP_Finite_LargeDataSize_LargeM}
\end{figure}

As we showed in Section \ref{Sec:SimImperfectIM}, using a very
bright source can improve the efficiencies of both passive and
active estimates. Here we again adjust the source intensity in Bob's
lab as $M_\mathrm{B}=10^{10}$. The results are shown in Figure
\ref{FIG:WV_GYS_ImPerfectIM_PnP_Finite_LargeDataSize_LargeM}. We can
see that using a very bright source can improve the efficiencies of
both passive and active estimates in finite data size case. As we
mentioned in Section \ref{Sec:SimImperfectIM}, such brightness
($10^{10}$ photons per pulse) is achievable with a pulsed laser diode
and a fibre laser amplifier. However, non-linear effects should be
carefully considered \cite{PassiveEstimate_NonLinear}.

In future studies, it would be worthwhile to incorporate the finite key length security analyses \cite{Security:Finite_Hayashi, Decoy:FiniteUpperBound, Decoy:Finite07, Security:Finite_Scarani, Security:FinitePostProcessing} in the key generation rate calculation.

\subsection{Simulating the Set-up in
\cite{ExpQKD:PKU_Untrusted}}\label{Sec:SimPKU}

\begin{figure}\center
  \includegraphics[width=9cm]{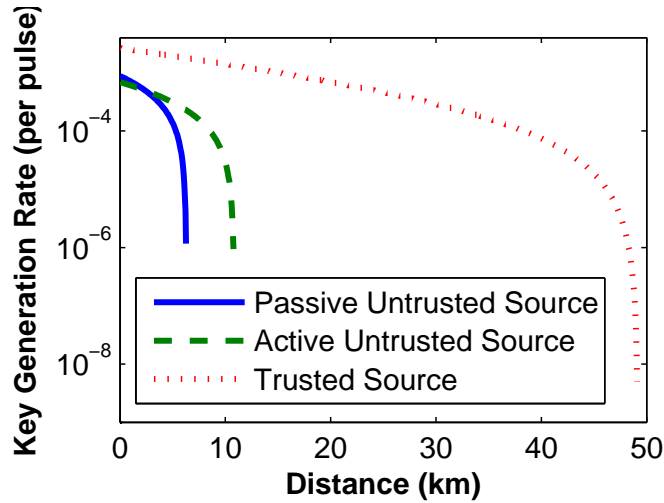}\\
  \caption{Simulation result of Weak+Vacuum \cite{Decoy:Practical}
  protocol based on the experimental parameters in \emph{\cite{ExpQKD:PKU_Untrusted}}: Data size is $9.05\times10^7$ (this data size is reported in \cite{ExpQKD:PKU_Untrusted}. It is smaller than the data size we assumed in other simulations. If a larger data size was used, we would expect some improvements on the simulation results), the intensity
  monitor efficiency $\eta_\mathrm{IM}=0.8$, the intensity
  monitor noise $\sigma_\mathrm{IM}=3.097\times10^5$, the intensity monitor
  conservative interval $\varsigma=6\sigma_\mathrm{IM}$, the source at Bob's side is
  Poissonian
  centered at $M_\mathrm{B}=6.411\times10^7$
  photons per pulse, the beam splitting ratio $q=0.05$,
  and the system is in Plug \& Play.
  Confidence level is set as
  $\tau\ge1-10^{-10}$. 6 standard deviations are considered in the statistical
  fluctuation. Single photon detector efficiency is $4\%$, detector error rate is
$1.39\%$, and background rate $Y_0=9.38\times10^{-5}$. Comparing
with Figure \ref{FIG:WV_GYS_ImPerfectIM_PnP_Finite_LargeDataSize},
we can see that higher background rate limits the system
performance. }\label{FIG:WV_PKU_ImPerfectIM_PnP}
\end{figure}

\cite{ExpQKD:PKU_Untrusted} reports so far the only experimental
implementation of QKD that considers the untrusted source
imperfection. However, as we discussed above, the analysis proposed
in \cite{ExpQKD:PKU_Untrusted} is challenging to use, and was not
applied to analyze the experimental results reported in the same
paper. Our analysis, however, provides a method to understand the
experimental results of \cite{ExpQKD:PKU_Untrusted}. Here, we
present a numerical simulation of the system used in
\cite{ExpQKD:PKU_Untrusted}.

We have to characterize the noise and conservative interval of the
intensity monitor used in \cite{ExpQKD:PKU_Untrusted}. The
experimental results reported in \cite{ExpQKD:PKU_Untrusted} show
that the \emph{measured} input photon number distribution is
centered at $M=1.818\times10^7$ with a standard deviation
$3.097\times10^5$ on Alice's side. If we assume the source at Bob's
side as Poissonian, the \emph{actual} input photon number
distribution on Alice's side will also be Poissonian. The detector
noise is then
$\sigma_\mathrm{IM}=\sqrt{(3.097\times10^5)^2-1.818\times10^7}
=3.097\times10^5$. We set the detector conservative interval as
constant $\varsigma=6\sigma_\mathrm{IM}$.

Source intensity at Bob's side $M_\mathrm{B}$ can be calculated in
the following manner: Since $M=1.818\times10^7$ at a distance $l=25$
km, and beam splitting ratio $q=0.05$, we can conclude that
\begin{equation*}
    M_\mathrm{B} = \frac{M}{\alpha l(1-q)}=6.411\times10^7.
\end{equation*}
Here we assume that the fibre loss coefficient $\alpha=-0.21$ dB/km.

The other parameters are directly cited from
\cite{ExpQKD:PKU_Untrusted}: The set-up is in Plug \& Play
structure. The efficiency of the intensity monitor is
$\eta_\mathrm{IM}=0.8$. Single photon detector efficiency is $4\%$,
detector error rate is $1.39\%$, and background rate
$Y_0=9.38\times10^{-5}$. As in previous sections, confidence level
is set as $\tau\ge1-10^{-10}$.

In the experiment reported in \cite{ExpQKD:PKU_Untrusted}, the data size
is $9.05\times10^7$ (it is smaller than the data size we assumed in other simulations. If a larger data size were used, we would expect some improvements on the simulation results). We ran numerical simulation with 6 standard
deviations that are considered in the statistical fluctuation. The
simulation results are shown in Figure
\ref{FIG:WV_PKU_ImPerfectIM_PnP}. It is encouraging to see that the
simulation yields positive key rates for both passive and active
estimates at short distances.

\subsection{Summary}

From the numerical simulations shown in Figures
\ref{FIG:GLLP_GYS_PerfectIM} -- \ref{FIG:WV_PKU_ImPerfectIM_PnP}, we
conclude that four important parameters can improve the efficiency
of passive estimate on an untrusted source: First, the beam
splitting ratio $q$ should be very small, say 1\%, to send most
input photons to the intensity monitor. Second, the light source
should be very bright (say, $10^{10}$ photons per pulse). This is
particularly important for Plug \& Play structure. Third, the
imperfections of the intensity monitor should be small. That is, the
intensity monitor should have high efficiency (say, over 70\%) and
high precision (say, can resolve photon number difference of
$6\times10^5$). Fourth, the data size should be large (say,
$10^{12}$ bits) to minimize the statistical fluctuation.

In brief, a largely biased beam splitter, a bright source, an efficient and
precise intensity monitor, and a large data size are four key
conditions that can substantially improve the efficiency of the passive
estimate on an untrusted source. The latter three conditions are
also applicable in the active estimate scheme.

An important advantage of decoy state protocols is that the key
generation rate will only drop linearly as channel transmittance
decreases
\cite{Decoy:Hwang,Decoy:LoISIT,Decoy:LoPRL,Decoy:Practical,
Decoy:WangPRL,Decoy:WangPRA,Decoy:ZhaoPRL,Decoy:ZhaoISIT}, while in
many non-decoy protocols, like the GLLP protocol \cite{GLLP}, the key
generation rate will drop quadratically as channel transmittance
decreases. In the simulations shown in Figures
\ref{FIG:WV_GYS_PerfectIM} -- 
\ref{FIG:WV_PKU_ImPerfectIM_PnP}, we can see that this important
advantage is preserved even if the source is unknown and untrusted.

\section{Preliminary Experimental Test}\label{Sec:PreExp}

We performed some preliminary experiments to test our analysis. The
basic idea is to measure some key parameters of our system,
especially the characteristics of the source, with which we can
perform numerical simulation to show the expected performance.

The experimental set-up is shown in Figure \ref{FIG:ExpSetup}. It is
essentially a modified commercial plug \& play QKD system. We added
a 1/99 beam splitter (1/99 BS in Figure \ref{FIG:ExpSetup}), a
photodiode (PD in Figure \ref{FIG:ExpSetup}), and a high-speed
oscilloscope (OSC in Figure \ref{FIG:ExpSetup}) on Alice's side.
These three parts comprise Alice's PNA.

\begin{figure}\center
  \includegraphics[width=10cm]{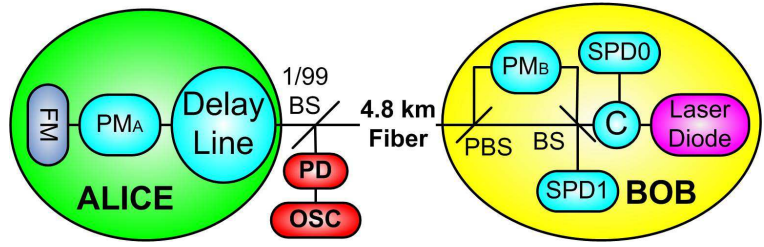}\\
  \caption{Experimental set-up. Alice and Bob: Commercial plug \& play
  QKD system. PD: photodiode. OSC: high-speed oscilloscope. 1/99 BS:
  1/99 beam splitter. FM: faraday mirror. PM$_\mathrm{x}$: phase modulators.
  PBS: polarizing beam splitter. BS: beam splitter. SPD$_\mathrm{x}$:
  single photon detector. C: circulator.}\label{FIG:ExpSetup}
\end{figure}

When Bob sends strong laser pulses to Alice, the photodiode (PD in
Figure \ref{FIG:ExpSetup}) will convert input photons into
photoelectrons, which are then recorded by the oscilloscope (OSC in
Figure \ref{FIG:ExpSetup}). In the recorded waveform, we calculated
the area below each pulse. This area is proportional to the number
of input photons. The conversion coefficient between the area and
photon number is calibrated by measuring the average input laser
power on Alice's side with a slow optical power meter.

In our experiment, 299 700 pulses are generated by the laser diode
at Bob's side (Laser Diode in Figure \ref{FIG:ExpSetup}) at a
repetition rate of 5 MHz with 1 ns pulse width. They are all
split into U pulses and L pulses (see Figure \ref{FIG:UL_Pulse})
by the 1/99 beam splitter (1/99 BS in Figure \ref{FIG:ExpSetup}).
The L pulses are measured by a photodiode (PD in Figure
\ref{FIG:ExpSetup}). The measurement results are acquired and
recorded by an oscilloscope (OSC in Figure \ref{FIG:ExpSetup}).

The experimental results of the photon number statistics are plotted
in Figure \ref{FIG:ExpDistribution}. The measured photon number
distribution is centered at $M=5.101\times10^6$ photons per pulse, with
standard deviation $6.557\times10^4$ on Alice's side. We can see
that the actual photon number distribution fits a Gaussian
distribution (shown as the blue line) well. Other experimental
results are shown in Table \ref{Tab:ExpParameter}.

\begin{table}[!b]\center
\caption{Parameters measured from our preliminary experiment
described in Section \ref{Sec:PreExp}.}\label{Tab:ExpParameter}
\begin{tabular}{cccc}
  \hline
   $\alpha$ & $\eta_\mathrm{det}$ & $e_\mathrm{det}$ & $Y_0$ \\
  \hline
  -0.21 dB/km& 4.89\% & 0.21\% & $8.4\times10^{-5}$ \\
  \hline
\end{tabular}
\end{table}

The intensity monitor noise is calculated in a similar manner to that 
in Section \ref{Sec:SimPKU}: Assuming the source is Poissonian at
Bob's side, which means the actual input photon number on Alice's
side is also Poissonian, the noise is then given by
$\sigma_\mathrm{IM}=\sqrt{(6.557\times10^4)^2-5.101\times10^6} =6.553\times 10^4.$
As in Section \ref{Sec:SimPKU}, we set the detector conservative
interval as a constant $\varsigma=6\sigma_\mathrm{IM}$.

Source intensity at Bob's side $M_\mathrm{B}$ can be calculated in
the following matter (which is similar to the one we used in Section
\ref{Sec:SimPKU}): Since $M=5.101\times10^6$ at a distance $l=4.8$
km, and beam splitting ratio $q=0.01$, we can conclude that
\begin{equation*}
    M_\mathrm{B} = \frac{M}{\alpha l(1-q)}=6.500\times10^6.
\end{equation*}
Here we know that the fibre loss coefficient $\alpha=-0.21$ dB/km.

\begin{figure}\center
  \includegraphics[width=10cm]{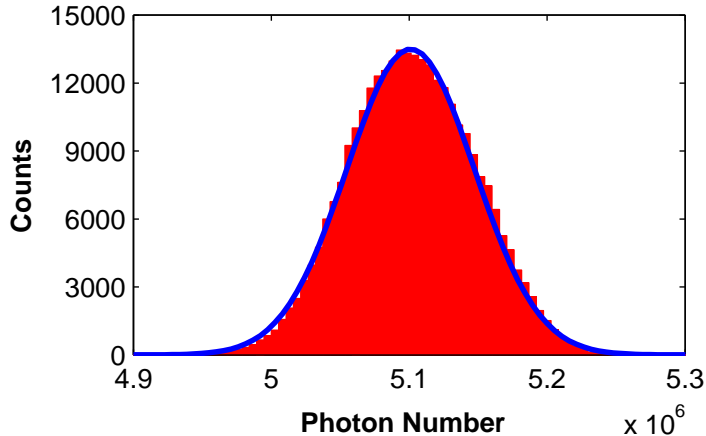}\\
  \caption{Experimentally measured photon number statistics for
  299 700 pulses. The distribution centered at $5.101\times10^6$
  photons per pulse, with standard deviation $6.557\times10^4$. Blue line shows
  a Gaussian fit of the actual distribution.}\label{FIG:ExpDistribution}
\end{figure}

The simulation result is shown in Figure
\ref{FIG:WV_Toronto_ImPerfectIM_PnP}, in which the data size is set
as $10^{12}$ \footnote{Data size in our experiment is much smaller
than the data size assumed in numerical simulation. The purpose of
our preliminary experiment is to test if it is possible to achieve
positive key rate with our current system.}. We can see that it is
possible to achieve positive key rate at moderate distances using
the security analysis presented in this paper.

\section{Conclusion}

In this paper, we present the first passive security analysis for
QKD with an untrusted source, with a complete security proof. Our
proposal is compatible with inefficient and noisy intensity
monitors, which is not considered in \cite{Security:UntrustedSource}
or in \cite{ExpQKD:PKU_Untrusted}. Our analysis is also compatible
with a finite data size, which is not considered in
\cite{ExpQKD:PKU_Untrusted}. Comparing to the active estimate scheme
proposed in \cite{Security:UntrustedSource}, the passive scheme
proposed in this paper significantly reduces the challenges to
implement the ``Plug \& Play'' QKD with unconditional security.
 Our proposal can
be applied to practical QKD set-ups with untrusted sources,
especially the plug \& play QKD set-ups, to guarantee the security.

\begin{figure}[!t]\center
  \includegraphics[width=9cm]{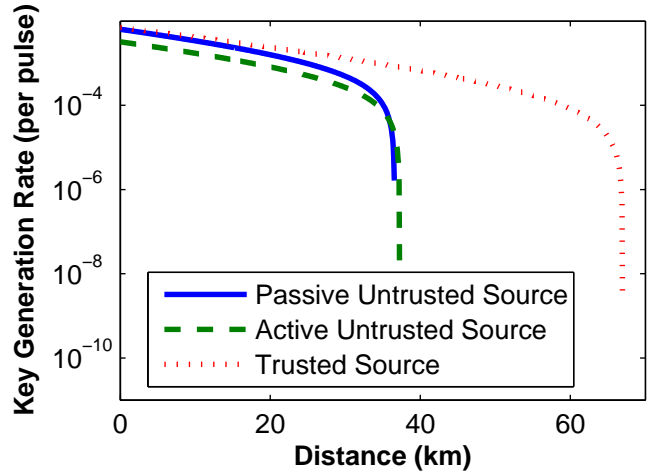}\\
  \caption{Simulation result of Weak+Vacuum \cite{Decoy:Practical}
  protocol based
  on experimental parameters from \emph{our QKD system}.
  We assume
  that the data size is $10^{12}$ bits, the intensity
  monitor efficiency $\eta_\mathrm{IM}=0.7$, the intensity
  monitor noise $\sigma_\mathrm{IM}=6.553\times10^4$, the intensity monitor
  conservative interval $\varsigma=6\sigma_\mathrm{IM}$, the source at Bob's lab is Poissonian
  centered at $M_\mathrm{B}=6.500\times10^6$
  photons per pulse, the beam splitting ratio $q=0.01$,
  and the system is in the Plug \& Play structure.
  Confidence level is set as
  $\tau\ge1-10^{-10}$. 6 standard deviations are considered in the statistical
  fluctuation.
  Experimental parameters are listed in Table \ref{Tab:ExpParameter}.
  }\label{FIG:WV_Toronto_ImPerfectIM_PnP}
\end{figure}

We point out four important conditions that can improve the
efficiency of the passive estimate scheme proposed in this paper:
First, the beam splitter in PNA should be largely biased to send
most photons to the intensity monitor. Second, the light source
should be bright. Third, the intensity monitor should have high
efficiency and precision. Fourth, the data size should be large to
minimize the statistical fluctuation. These four conditions are
confirmed in extensive numerical simulations.

In the simulations shown in Figures \ref{FIG:WV_GYS_ImPerfectIM_PnP}
--  \ref{FIG:WV_PKU_ImPerfectIM_PnP} and Figure
\ref{FIG:WV_Toronto_ImPerfectIM_PnP}, we made an additional
assumption that the intensity monitor has a constant Gaussian noise.
This assumption is \emph{not} required by our security analysis. It
will be interesting to experimentally verify this model in future.

The numerical simulations show that if the above conditions are met,
the efficiency of the passive untrusted source estimate is close to
that of the trusted source estimate, and is roughly twice as high as the
efficiency of the active untrusted source estimate. Nonetheless, the
efficiency of active estimate scheme proposed in
\cite{Security:UntrustedSource} may be improved to the level that is
similar to the efficiency of passive estimation. This is briefly discussed below Equation \eqref{Eq:SamplingCodingDeviation_Proof}. The security of
the improved active estimate scheme is beyond the scope of the current
paper, and is subject to further investigation.

Numerical simulations in Figures \ref{FIG:WV_GYS_PerfectIM} --
\ref{FIG:WV_PKU_ImPerfectIM_PnP} and Figure
\ref{FIG:WV_Toronto_ImPerfectIM_PnP} show that the key generation
rate drops linearly as the channel transmittance decreases. This is
an important advantage of decoy state protocols over many other QKD
protocols, and is preserved in our untrusted source analysis.

Our preliminary experimental test highlights the feasibility of our
proposed passive estimate scheme. Indeed, our scheme can be easily
implemented by making very simple modifications (by adding a few
commercial modules) to a commercial Plug \& Play QKD system.

A remaining practical question in our proposal is: How to calibrate
the noise and the conservative interval of the intensity monitor?
Note that these two parameters may not be constant at different
intensity levels. Moreover, the noise may not be Gaussian. It is not
straightforward to define the conservative interval and its
confidence.

\ack

The authors thank Lei Wu and Marc Napoleon for performing some
experimental tests, which is supported by NSERC USRA. The authors
thank Jean-Christian Boileau for discussions at the early stage of
this work. Support of the funding agencies CFI, CIPI, the CRC
program, CIFAR, MITACS, NSERC, OIT, PREA, and Quantum Works is
gratefully acknowledged.

\begin{appendix}

\section{The Phase Randomization Assumption}\label{App:PR}

In this section, we will show that for the state that is accessible to Eve, Alice's phase randomization is equivalent to performing a QND measurement on the photon number of the input pulses.

\begin{proof}
Before the phase randomization, the state that is shared by Alice, Bob, and Eve is
\begin{equation}\label{EQ:StateABE}
\ket{\psi}_\text{ABE}=\sum_{m,n}b_{m,n}\ket{E_n}\ket{m}.
\end{equation}
After the phase randomization, a random phase is applied to each pulse, and this phase is inaccessible to Eve. The state becomes
\begin{equation}\label{EQ:StateABE_PR}
\ket{\psi}_\text{ABE}^\text{pr}=\sum_{m,n}b_{m,n}\ket{E_n}\ket{m}e^{im\theta}.
\end{equation}
Its density matrix is
\begin{equation}\label{EQ:DM_ABE_PR}
\rho_\text{ABE}^\text{pr}=\sum_{m,n,m',n'}b_{m,n}b^*_{m',n'}\ket{E_n}\bra{E_{n'}} \otimes \ket{m} \bra{m'} e^{i(m-m')\theta}.
\end{equation}
Since $\theta$ is not known to Eve or Bob, the state that is accessible to Bob and Eve is
\begin{equation}\label{EQ:DM_BE_PR}
\begin{aligned}
\rho_\text{BE}^\text{pr} &= \frac{1}{2\pi} \sum_{m,n,m',n'}b_{m,n}b^*_{m',n'}\ket{E_n}\bra{E_{n'}} \otimes \ket{m} \bra{m'} \int_0^{2\pi} e^{i(m-m')\theta} \text{d}\theta \\
&= \sum_{m,n,n'}b_{m,n}b^*_{m,n'}\ket{E_n}\bra{E_{n'}} \otimes \ket{m} \bra{m}.
\end{aligned}
\end{equation}

Instead of considering the phase randomization, now let us analyze the impact of a photon number measurement on the state given in Equation \eqref{EQ:StateABE}. Before Alice's measurement, the density matrix is given by

\begin{equation}\label{EQ:DM_ABE}
\rho_\text{ABE}=\sum_{m,n,m',n'}b_{m,n}b^*_{m',n'}\ket{E_n}\bra{E_{n'}} \otimes \ket{m} \bra{m'}.
\end{equation}

After the QND measurement of the photon number, Alice knows the photon number. However, Eve and Bob do not know Alice's measurement result. Therefore, the state that is accessible to Bob and Eve is given by

\begin{equation}\label{EQ:DM_BE_QND}
\begin{aligned}
\rho_\text{BE}^\text{QND} &= \sum_{m''}\ket{m''}\bra{m''}\rho_\text{ABE}\ket{m''}\bra{m''}\\
&= \sum_{m,n,n'}b_{m,n}b^*_{m,n'}\ket{E_n}\bra{E_{n'}} \otimes \ket{m} \bra{m}\\
&=\rho_\text{BE}^\text{pr}.
\end{aligned}
\end{equation}
\end{proof}

From the above equation, we conclude that the two different processes---i) phase randomization by Alice and ii) a photon-number non-demolition measurement---actually give exactly the same density matrix for the Eve-Bob system. Therefore, phase randomization is mathematically equivalent to a photon-number non-demolition measurement. For this reason, we can consider the output state by Alice as a classical mixture of Fock states.

\section{Security Analysis for Untagged Bits}\label{App:Security}

In this section, for the convenience of the readers, we  recapitulate the security analysis that is presented in \cite{Security:UntrustedSource}.

Assume that $k$ pulses are sent from Alice to Bob. Alice and Bob do not know which bits are untagged. However, either from the active estimate presented in \cite{Security:UntrustedSource} or the passive estimate presented in the current paper, they know that at least $(1-\Delta-\epsilon) k$ pulses are untagged with high confidence.

Alice and Bob can measure the overall gain $Q_e$ and the overall QBER $E_e$. They do not know the gain $Q$ and the QBER $E$ for the untagged bits because they do not know which bits are untagged. Nonetheless, they can then estimate the upper
bounds and the lower bounds of them. The upper bound and the lower bound of
$Q$ are \cite{Security:UntrustedSource}
\begin{equation}\label{eq:Qbound}
\begin{aligned}
    \overline{Q} &= \frac{Q_e}{1-\Delta-\epsilon},\\
    \underline{Q} &= \max(0,
    \frac{Q_e-\Delta-\epsilon}{1-\Delta-\epsilon}).
\end{aligned}
\end{equation}
The upper bound and lower bound of $E\cdot Q$ can be estimated as \cite{Security:UntrustedSource}
\begin{equation}\label{eq:Ebound}
    \begin{aligned}
    \overline{E\cdot Q} &=
    \frac{Q_e E_e}{1-\Delta-\epsilon},\\
    \underline{E\cdot Q} &=
    \max(0,\frac{Q_e E_e-\Delta-\epsilon}{1-\Delta-\epsilon}).
    \end{aligned}
\end{equation}

For untagged bits (i.e., $m\in[(1-\delta) M, (1+\delta) M]$), we can
show that the upper bound and the lower bound of the probability that the output photon number from Alice is $n$ are \cite{Security:UntrustedSource}:
\begin{equation}\label{eq:pnbound}
\begin{aligned}
    \overline{P_n} &= \left\{
                        \begin{array}{ll}
                         (1-\lambda)^{(1-\delta)M}, & \hbox{if
                         $n=0$;}\\
                          {(1+\delta)M\choose n}\lambda^n(1-\lambda)^{(1+\delta)M-n}, & \hbox{if $1\le n \le (1+\delta)M$;} \\
                          0, & \hbox{if $n>(1+\delta)M$;}
                        \end{array}
                      \right.\\
    \underline{P_n} &= \left\{
                        \begin{array}{ll}
                        (1-\lambda)^{(1+\delta)M}, & \hbox{if
                        $n=0$;}\\
                          {(1-\delta)M\choose n}\lambda^n(1-\lambda)^{(1-\delta)M-n}, & \hbox{if $1\le n \le (1-\delta)M$;} \\
                          0, & \hbox{if $n>(1-\delta)M$;}
                        \end{array}
                      \right.
\end{aligned}
\end{equation}
under \textbf{Condition 1:}
\begin{equation}\label{eq:condition1}
(1+\delta)M\lambda<1.
\end{equation}

The key rate calculation depends on the QKD protocol that is implemented. For the GLLP \cite{GLLP} protocol with an untrusted source, the key generation rate is given by \cite{Security:UntrustedSource}:
\begin{equation}\label{eq:rateGLLP_Generalized}
R \ge
\frac{1}{2}\{-Q_ef(E_e)H_2(E_e)+(\underline{Q}+\underline{P_0}+\overline{P_1}-1)[1-H_2(\frac{Q_eE_e}{\underline{Q}+\underline{P_0}+\overline{P_1}-1})]\},
\end{equation}
where $Q_e$ and $E_e$ are measured experimentally, $\underline{Q}$ can be calculated from Equation \eqref{eq:Qbound}, and $\underline{P_0}$ and $\overline{P_1}$ can be calculated from Equation (\ref{eq:pnbound}).

For decoy state protocols \cite{Decoy:Hwang,Decoy:LoISIT,Decoy:LoPRL,Decoy:Practical,Decoy:WangPRL, Decoy:WangPRA, Decoy:ZhaoPRL, Decoy:ZhaoISIT}, the key generation rate (with an untrusted source) is given by \cite{Security:UntrustedSource}:
\begin{equation}\label{eq:rateWV}
R \ge 
\frac{1}{2}\{-Q_{e}^Sf(E_{e}^S)H_2(E_{e}^S)+(1-\Delta-\epsilon)\underline{Q_1^S}[1-H_2(\overline{e_1^S})]\}
,
\end{equation}
where $Q_{e}^S$ and $E_{e}^S$ are the overall gain and the over QBER of the signal states, respectively, and can be measured experimentally. $\underline{Q_1^S}$ and $\overline{e_1^S}$ depend on the specific decoy state protocol that is implemented. 

For weak+vacuum protocol \cite{Decoy:Practical,Decoy:WangPRL,Decoy:ZhaoISIT}, the lower bound of $Q_1^S$ for untagged bits is given by \cite{Security:UntrustedSource}:
\begin{equation}\label{eq:q1boundWV}
    \begin{aligned}
\underline{Q_1^S}= \underline{P_1^S}\frac{\underline{Q^D}\underline{P_2^S}-\overline{Q^S}\overline{P_2^D}+(\underline{P_0^S}\overline{P_2^D}-\overline{P_0^D}\underline{P_2^S})\overline{Q^V}-\frac{2\delta
   M(1-\lambda_D)^{2\delta
   M-1}\underline{P_2^S}}{[(1-\delta)M+1]!}}{\overline{P_1^D}\underline{P_2^S}-\underline{P_1^S}\overline{P_2^D}}
    \end{aligned}
\end{equation}
under \textbf{Condition 2:}
\begin{equation}
\frac{\lambda_S}{\lambda_D}>\frac{(1+\delta)M-2}{(1-\delta)M-2}\left[\frac{(1+\delta)M-2}{2\delta
M}\right]^\frac{2\delta
M}{(1-\delta)M-2}\left[\frac{(1+\delta)M-2}{(1-\delta)M-2}\cdot\frac{e^2}{2\delta
M}\right]^\frac{1}{2[(1-\delta)M-2]}.
\end{equation}
Here $Q^S$, $Q^D$ and $Q^V$ are the gains of untagged bits of the
signal state, the decoy state, and the vacuum state, respectively.
Their bounds can be estimated from Equations \eqref{eq:Qbound}. The
bounds of the probabilities can be estimated from Equations
\eqref{eq:pnbound}. $\lambda_S$ and $\lambda_D$ are Alice's internal transmittances for signal and decoy states, respectively. The upper bound of $e_1^S$ for untagged bits
 is given by \cite{Security:UntrustedSource}:
\begin{equation}\label{eq:e1boundWV}    e_1^S\le\overline{e_1^S}=\frac{\overline{E^SQ^S}-\underline{P_0^S}\underline{E^VQ^V}}{\underline{Q_1^S}},
\end{equation}
in which $E^S$ and $E^V$ are the QBERs of untagged bits of the
signal and the vacuum states, respectively. $\overline{E^SQ^S}$ and
$\underline{E^VQ^V}$ can be estimated from Equations \eqref{eq:Ebound}.
$\underline{P_0^S}$ can be estimated  by Equations \eqref{eq:pnbound}.
$\underline{Q_1^S}$ is given by Equation \eqref{eq:q1boundWV}.

For one-decoy protocol \cite{Decoy:Practical,Decoy:ZhaoPRL}, a lower bound of $Q_1^S$ and
an upper bound of $e_1^S$ for untagged bits are given by
\begin{equation}\label{eq:q1e1_OneDecoy}
    \begin{aligned}
    \underline{Q_1^S} &= \underline{P_1^S}\frac{\underline{Q^D}\underline{P_2^S}-\overline{Q^S}\overline{P_2^D}+(\underline{P_0^S}\overline{P_2^D}-\overline{P_0^D}\underline{P_2^S})\frac{\overline{E^SQ^S}}{\underline{P_0^S}E^V}-\frac{2\delta
   M(1-\lambda_D)^{2\delta
   M-1}\underline{P_2^S}}{[(1-\delta)M+1]!}}{\overline{P_1^D}\underline{P_2^S}-\underline{P_1^S}\overline{P_2^D}},\\
    \overline{e_1^S} &= \frac{\overline{E^S\cdot
    Q^S}}{\underline{Q_1^S}},
    \end{aligned}
\end{equation}
respectively, under Condition 2 in the asymptotic case. Here $Q^S$
and $Q^D$ are the gains of untagged bits of the signal state and the
decoy state, respectively. Their bounds can be estimated from Equations
\eqref{eq:Qbound}. $E^S$ is the QBER of untagged bits of the signal
state. $\overline{E^S\cdot Q^S}$ can be estimated from Equations
\eqref{eq:Ebound}. $E^V=0.5$ in the asymptotic case. The bounds of
the probabilities can be estimated from Equations \eqref{eq:pnbound}.

\section{Confidence Level in Active
Estimate}\label{App:Confidence_Active}

Among all the $V$ untagged bits, each bit has probability 1/2 to be
assigned as an untagged coding bit. Therefore, the probability that
$V_\mathrm{c}=v_\mathrm{c}$ obeys a binomial distribution. Cumulative
probability is given by \cite{LargeDeviation_Hoeffding}
\begin{equation*}
    P(V_\mathrm{c}\le\frac{V-\epsilon k}{2}
    |V=v)
    \le \exp(-\frac{\epsilon^2k^2}{2v})
\end{equation*}

For any $v\in[0,k]$, $k/v\ge 1$. Therefore, we have

\begin{equation*}
    P(V_\mathrm{c}\le\frac{V-\epsilon k}{2}|V\in[0, k]) \le \exp(-\frac{k\epsilon^2}{2}).
\end{equation*}

In the experiment described by Lemma
\ref{Lemma:ActiveConfidenceLevel}, $V\in[0,k]$ is always true.
Therefore, the above inequality reduces to

\begin{equation}\label{Eq:CumulativeBinomial}
    P(V_\mathrm{c}\le\frac{V-\epsilon k}{2}) \le \exp(-\frac{k\epsilon^2}{2}).
\end{equation}

By definition, we have
\begin{equation}\label{Eq:SumUntagged}
V=V_\mathrm{c}+V_\mathrm{s}.
\end{equation}
Substituting Equation \eqref{Eq:SumUntagged} into Equation
\eqref{Eq:CumulativeBinomial}, we have
\begin{equation}\label{Eq:SamplingCodingDeviation_Proof}
    P(V_\mathrm{c}\le V_\mathrm{s}-\epsilon k)\le
    \exp(-\frac{k\epsilon^2}{2}).\qed
\end{equation}

The above proof can be easily generalized to the case where for each
bit sent from the untrusted source to Alice, Alice randomly assigns
it as either a coding bit with probability $\gamma$, or a sampling
bit with probability $1-\gamma$. Here $\gamma\in(0,1)$ is chosen by
Alice. It is then straightforward to show that

\begin{equation}\label{Eq:SamplingCodingDeviation_General}
    P[V_\mathrm{c}\le \frac{\gamma}{1-\gamma}(V_\mathrm{s}-\epsilon k)]
    \le \exp(-2k\epsilon^2\gamma^2).
\end{equation}

When $\gamma=1/2$, Equation
\eqref{Eq:SamplingCodingDeviation_General} reduces to Equation
\eqref{Eq:SamplingCodingDeviation_Proof}.

\section{Confidence Level in Cross
Estimate}\label{App:Confidence_Cross}

From Corollary \ref{Col:1} and Corollary \ref{Col:2}, we know that

\begin{equation}\label{EQ:Individual_Confidence}
    \begin{aligned}
        P(V_\mathrm{c}^\mathrm{U} &\le V_\mathrm{s}^\mathrm{L} - \epsilon_1 k
        ) \le \exp(\frac{-k\epsilon_1^2}{2})\\
        P(V_\mathrm{s}^\mathrm{U} &\le V_\mathrm{c}^\mathrm{L} - \epsilon_2 k
        ) \le \exp(\frac{-k\epsilon_2^2}{2}).
    \end{aligned}
\end{equation}
Therefore, we have
\begin{equation}\label{EQ:Combined_Confidence}
    \begin{aligned}
        & P(V^\mathrm{U}\le V_\mathrm{s}^\mathrm{L}+V_\mathrm{c}^\mathrm{L}-
            (\epsilon_1 + \epsilon_2) k)\\
        = & P(V_\mathrm{c}^\mathrm{U}+V_\mathrm{s}^\mathrm{U}\le V_\mathrm{s}^\mathrm{L}+V_\mathrm{c}^\mathrm{L}-
            (\epsilon_1 + \epsilon_2) k)\\
        \le & P[ (V_\mathrm{c}^\mathrm{U} \le V_\mathrm{s}^\mathrm{L} - \epsilon_1
        k) \\
        &\text{ or } (V_\mathrm{s}^\mathrm{U} \le V_\mathrm{c}^\mathrm{L} - \epsilon_2 k)
        ]\\
        \le & P(V_\mathrm{c}^\mathrm{U} \le V_\mathrm{s}^\mathrm{L} - \epsilon_1
        k)\\
        &+ P(V_\mathrm{c}^\mathrm{U} \le V_\mathrm{s}^\mathrm{L} -
        \epsilon_2
        k)\\
        =& \exp(\frac{-k\epsilon_1^2}{2}) +
\exp(\frac{-k\epsilon_2^2}{2}).
    \end{aligned}
\end{equation}
In the above derivation, we made use of the fact that
$[(V_\mathrm{c}^\mathrm{U} \le V_\mathrm{s}^\mathrm{L} - \epsilon_1
        k) \text{ or } (V_\mathrm{s}^\mathrm{U} \le V_\mathrm{c}^\mathrm{L} - \epsilon_2
        k)]$ is always true if $[V_\mathrm{c}^\mathrm{U}+V_\mathrm{s}^\mathrm{U}\le V_\mathrm{s}^\mathrm{L}+V_\mathrm{c}^\mathrm{L}-
            (\epsilon_1 + \epsilon_2) k]$ is true.

\end{appendix}

\section*{References}
\bibliography{../../reference}

\begin{thebibliography}{10}
\newcommand{\enquote}[1]{``#1''}
\expandafter\ifx\csname url\endcsname\relax
  \def\url#1{{#1}}\fi
\expandafter\ifx\csname urlprefix\endcsname\relax\def\urlprefix{}\fi

\bibitem{BB84}
C.~H. Bennett and G.~Brassard, \enquote{Quantum Cryptography: Public Key
  Distribution and Coin Tossing,} in {\em Proceedings of IEEE International
  Conference on Computers, Systems, and Signal Processing\/}  pp. 175 -- 179
  (1984).

\bibitem{Ekert91}
A.~K. Ekert, \enquote{Quantum cryptography based on {Bell}'s theorem,} Phys.
  Rev. Lett. {\bf 67}, 661 (1991).

\bibitem{Review:Encyclopedia}
H.-K. Lo and Y.~Zhao, \enquote{Quantum Cryptography,} in {\em Encyclopedia of
  Complexity and System Science\/} (Springer, New York, 2009), Vol.~8, pp.
  7265--7289, arXiv:0803.2507.

\bibitem{SecurityProofs}
D. Mayers, J. of ACM \textbf{48}, 351 (2001); H.-K. Lo and H. F. Chau, Science
  \textbf{283}, 2050 (1999); P. Shor and J. Preskill, Phys. Rev. Lett.
  \textbf{85}, 441 (2000).

\bibitem{GLLP}
D.~Gottesman, H.-K. Lo, N.~L\protect{\"{u}}tkenhaus, and J.~Preskill,
  \enquote{Security of quantum key distribution with imperfect devices,} Quant.
  Info. Compu. {\bf 4}, 325 (2004).

\bibitem{ILM}
H.~Inamori, N.~L\protect{\"{u}}tkenhaus, and D.~Mayers, \enquote{Unconditional
  Security of Practical Quantum Key Distribution,} European Physical Journal D
  {\bf 41}, 599 (2007).

\bibitem{Decoy:Hwang}
W.~Y. Hwang, \enquote{Quantum Key Distribution with High Loss: Toward Global
  Secure Communication,} Phys. Rev. Lett. {\bf 91}, 057\,901 (2003).

\bibitem{Decoy:LoISIT}
H.-K. Lo, \enquote{Quantum Key Distribution with Vacua or Dim Pulses as Decoy
  States,} in {\em Proceedings of IEEE International Symposium on Information
  Theory\/}  p. 137 (2004).

\bibitem{Decoy:LoPRL}
H.-K. Lo, X.~Ma, and K.~Chen, \enquote{Decoy State Quantum Key Distribution,}
  Phys. Rev. Lett. {\bf 94}, 230\,504 (2005).

\bibitem{Decoy:Practical}
X.~Ma, B.~Qi, Y.~Zhao, and H.-K. Lo, \enquote{Practical Decoy State for Quantum
  Key Distribution,} Phys. Rev. A {\bf 72}, 012\,326 (2005).

\bibitem{Decoy:WangPRL}
X.-B. Wang, \enquote{Beating the Photon-Number-Splitting Attack in Practical
  Quantum Cryptography,} Phys. Rev. Lett. {\bf 94}, 230\,503 (2005).

\bibitem{Decoy:WangPRA}
X.-B. Wang, \enquote{Decoy-state protocol for quantum cryptography with four
  different intensities of coherent light,} Phys. Rev. A {\bf 72}, 012\,322
  (2005).

\bibitem{Decoy:ZhaoPRL}
Y.~Zhao, B.~Qi, X.~Ma, H.-K. Lo, and L.~Qian, \enquote{Experimental Quantum Key
  Distribution with Decoy States,} Phys. Rev. Lett. {\bf 96}, 070\,502 (2006).

\bibitem{Decoy:ZhaoISIT}
Y.~Zhao, B.~Qi, X.~Ma, H.-K. Lo, and L.~Qian, \enquote{Simulation and
  Implementation of Decoy State Quantum Key Distribution over 60km Telecom
  Fiber,} in {\em Proceedings of IEEE International Symposium of Information
  Theory\/}  pp. 2094--2098 (2006).

\bibitem{Proposal:PnP}
A.~Muller, T.~Herzog, B.~Hutter, W.~Tittel, H.~Zbinden, and N.~Gisin,
  \enquote{`Plug \& play' systems for quantum crytography,} Appl. Phys. Lett.
  {\bf 70}, 793 (1997).

\bibitem{ExpQKD:PnP_67km}
D.~Stucki, N.~Gisin, O.~Guinnard, G.~Ribordy, and H.~Zbinden, \enquote{Quantum
  key distribution over 67 km with a plug\&play system,} New J. of Phys. {\bf
  4}, 41 (2002).

\bibitem{IDQ}
{www}.idquantique.com.

\bibitem{MagiQ}
{www}.magiqtech.com.

\bibitem{CentralLimit}
G.~P\'{o}lya, \enquote{\"{U}ber den zentralen Grenzwertsatz der
  Wahrscheinlichkeitsrechnung und das Momentenproblem (in German),}
  Mathematische Zeitschrift {\bf 8}, 171 (1920).

\bibitem{ThrHack:TrojanHorse}
N.~Gisin, S.~Fasel, B.~Kraus, H.~Zbinden, and G.~Ribordy, \enquote{Trojan horse
  attacks on quantum key distribution systems,} Phys. Rev. A {\bf 73}, 022\,320
  (2006).

\bibitem{Security:UntrustedSource}
Y.~Zhao, B.~Qi, and H.-K. Lo, \enquote{Quantum Key Distribution with an Unknown
  and Untrusted Source,} Phys. Rev. A {\bf 77}, 052\,327 (2008).

\bibitem{ExpQKD:Ground_Satellite_Feasibility}
P.~Villoresi, T.~Jennewein, F.~Tamburini, M.~Aspelmeyer, C.~Bonato, R.~Ursin,
  C.~Pernechele, V.~Luceri, G.~Bianco, A.~Zeilinger, and C.~Barbieri,
  \enquote{Experimental verification of the feasibility of a quantum channel
  between Space and Earth,} New J. of Phys. {\bf 10}, 033\,038 (2008).

\bibitem{Decoy:WangAPL}
X.-B. Wang, C.-Z. Peng, and J.-W. Pan, \enquote{Simple protocol for secure
  decoy-state quantum key distribution with a loosely controlled source,} Appl.
  Phys. Lett. {\bf 90}, 031\,110 (2007).

\bibitem{Decoy:WangInexactSource}
X.-B. Wang, C.-Z. Peng, J.~Zhang, and J.-W. Pan, \enquote{Security of
  decoy-state quantum key distribution with inexactly controlled source,}
  quant-ph/0612121v3  (2008).

\bibitem{Decoy:WangGeneralErrorSource}
X.-B. Wang, C.-Z. Peng, J.~Zhang, and J.-W. Pan, \enquote{General theory of
  decoy-state quantum cryptography with source errors,} Phys. Rev. A {\bf 77},
  043\,311 (2008).

\bibitem{Decoy:WangNJP}
X.-B. Wang, L.~Yang, C.-Z. Peng, and J.-W. Pan, \enquote{Decoy-state quantum
  key distribution with both source errors and statistical fluctuations,} New
  J. of Phys. {\bf 11}, 075\,006 (2009).

\bibitem{DPSK:200km}
H.~Takesue, S.~W. Nam, Q.~Zhang, R.~H. Hadfield, T.~Honjo, K.~Tamaki, and
  Y.~Yamamoto, \enquote{Quantum key distribution over a 40-dB channel loss
  using superconducting single-photon detectors,} Nature Photonics {\bf 1}, 343
  (2007).

\bibitem{ExpQKD:PKU_Untrusted}
X.~Peng, H.~Jiang, B.~Xu, X.~Ma, and H.~Guo, \enquote{Experimental quantum key
  distribution with an untrusted source,} Opt. Lett. {\bf 33}, 2077 (2008).

\bibitem{GYS}
C.~Gobby, Z.~L. Yuan, and A.~J. Shields, \enquote{Quantum key distribution over
  122 km of standard telecom fiber,} Appl. Phys. Lett. {\bf 84}, 3762 (2004).

\bibitem{Gain}
The gain is defined to be the ratio of the number of receiver Bob's detection
  events to the number of signals emitted by sender Alice in the cases where
  Alice and Bob use the same basis. It depends mainly on the intensity of
  signal, channel transmittance, and Bob's quantum efficiency.

\bibitem{PassiveEstimate_NonLinear}
See, like, Gerd Keiser, Optical Fiber Communications, 3rd edition, Chapter 12.5
  (McGraw-Hill, 2000).

\bibitem{ExpFastQKD:MbpsYamamoto}
Q.~Zhang, H.~Takesue, T.~Honjo, K.~Wen, T.~Hirohata, M.~Suyama, Y.~Takiguchi,
  H.~Kamada, Y.~Tokura, O.~Tadanaga, Y.~Nishida, M.~Asobe, and Y.~Yamamoto,
  \enquote{Megabits secure key rate quantum key distribution,} New J. of Phys.
  {\bf 11}, 045\,010 (2009).

\bibitem{ExpFastQKD:MbpsShields}
A.~R. Dixon, Z.~L. Yuan, J.~F. Dynes, A.~W. Sharpe, and A.~J. Shields,
  \enquote{Gigahertz decoy quantum key distribution with 1 Mbit/s secure key
  rate,} Opt. Express {\bf 16}, 18\,790 (2008).

\bibitem{Security:Finite_Hayashi}
M.~Hayashi, \enquote{Practical evaluation of security for quantum key
  distribution,} Phys. Rev. A {\bf 74}, 022\,306 (2006).

\bibitem{Decoy:FiniteUpperBound}
M.~Hayashi, \enquote{Upper bounds of eavesdropper's performances in
  finite-length code with decoy method,} Phys. Rev. A {\bf 76}, 012\,329
  (2007).

\bibitem{Decoy:Finite07}
J.~Hasegawa, M.~Hayashi, T.~Hiroshima, and A.~Tomita, \enquote{Security
  analysis of decoy state quantum key distribution incorporating finite
  statistics,} arXiv:0707.3541  (2007).

\bibitem{Security:Finite_Scarani}
V.~Scarani and R.~Renner, \enquote{Quantum Cryptography with Finite Resources:
  Unconditional Security Bound for Discrete-Variable Protocols with One-Way
  Postprocessing,} Phys. Rev. Lett. {\bf 100}, 200\,501 (2008).

\bibitem{Security:FinitePostProcessing}
X.~Ma, C.-H.~F. Fung, J.-C. Boileau, and H.~F. Chau, \enquote{Practical
  post-processing for quantum-key-distribution experiments,} arXiv:0904.1994
  (2009).

\bibitem{LargeDeviation_Hoeffding}
W.~Hoeffding, \enquote{Probability inequalities for sums of bounded random
  variables,} J. Am. Stat. Asso. {\bf 58}, 13 (1963).

\end{thebibliography}
\bibliographystyle{osa}

\end{document}